\documentclass[11pt,a4paper]{article}
\usepackage{jheppub}
\usepackage[utf8]{inputenc}

\usepackage{amsmath,amssymb,amsthm}
\theoremstyle{plain}
\newtheorem{theorem}{Theorem}
\newtheorem{lemma}{Lemma}
\theoremstyle{definition}
\newtheorem{definition}{Definition}
\numberwithin{equation}{section}

\usepackage{xcolor}
\definecolor{darkgreen}{rgb}{0,.3,0}
\definecolor{darkblue}{rgb}{0,0,.5}
\definecolor{darkred}{rgb}{.4,0,0}

\usepackage[hyphens]{xurl}
\usepackage{hyperref} 

\hypersetup{
pdfauthor={Paul-Hermann Balduf, Davide Gaiotto},			 
pdfcreator={LaTeX with hyperref},
colorlinks=true,	 
linkcolor=darkred,	  
citecolor=darkgreen,	 
filecolor=black,	 
urlcolor=darkblue,	 
bookmarks=true,
plainpages=false,
pdfpagelabels=true,
breaklinks=true
}

\usepackage[nameinlink,noabbrev 
]{cleveref}	
\crefname{figure}{figure}{figures}

\usepackage[font=small,labelfont=bf]{caption}
\usepackage[font=footnotesize]{subcaption}
\usepackage{placeins} 
\usepackage{csquotes} 

\usepackage{textcomp}

\usepackage{graphicx}
\usepackage{tikz}
\usetikzlibrary{arrows.meta}
\usetikzlibrary{calc}
\tikzstyle{edge} = [black,line width=.4mm]
\tikzstyle{vertex}=[circle,minimum size=2.5mm, draw=black, fill=black, inner sep=0mm]

\usepackage[framemethod=tikz,
roundcorner=2pt,
linewidth=.5pt,
skipabove=10pt plus 20pt minus 2pt,
skipbelow=12pt plus 20pt minus 2pt]{mdframed}
\mdfsetup{linewidth=1.0pt, roundcorner=5pt}

\newcounter{example}

\NewDocumentEnvironment {example} {m}
{ \refstepcounter{example} 
	\pagebreak[2]  
  \begin{mdframed}[linewidth=0.8pt,  linecolor=black, 
		bottomline=false,topline=false,rightline=false, startcode=\needspace{4\baselineskip}] 
    \noindent \textbf{Example  \theexample.} ~  }
	{ \end{mdframed}  
}

\newcommand{\laplacian}{\mathbb{L}}
\newcommand{\incidencematrix}{\mathbb{I}}
\newcommand{\edgematrix}{\mathbb{D}}
\newcommand{\abs}[1]{\left | #1 \right |}
\renewcommand{\d}{ \,\textnormal{d}}
\newcommand{\sgn}{\operatorname{sgn}}
\newcommand{\fey} { \mathcal {F}   }
\newcommand{\feyk}[1] { \fey \left( #1 \right)  }

\newcommand{\vv}{ }

\title{Combinatorial proof of a Non-Renormalization Theorem}

\author[1,2]{Paul-Hermann Balduf,}
\author[2]{Davide Gaiotto}

\affiliation[1]{Mathematical Institute, University of Oxford, \\
	Andrew Wiles Building, Woodstock Road, Oxford, OX2 6GG, United Kingdom.}
\affiliation[2]{ Perimeter Institute for Theoretical Physics, \\
	31 Caroline Street, Waterloo, Ontario, N2L 2Y5, Canada.}

\emailAdd{paul-hermann.balduf@maths.ox.ac.uk}
\emailAdd{dgaiotto@perimeterinstitute.ca}

\arxivnumber{2408.03192}

\abstract{
We provide a direct combinatorial proof of a Feynman graph identity which implies a wide generalization of a formality theorem by Kontsevich. For a Feynman graph $\Gamma$, we associate to each vertex a position $x_v \in \mathbb R$ and to each edge $e$ the combination $s_e = a_e^{-\frac 12} \left( x^+_e - x^-_e \right)$, where $x^\pm_e$ are the positions of the two end vertices of $e$, and $a_e$ is a Schwinger parameter. The ``topological propagator'' $P_e = e^{-s_e^2}\d s_e$ includes a part proportional to $\d x_v$ and a part proportional to $\d a_e$. Integrating the product of all $P_e$ over positions produces a differential form $\alpha_\Gamma$ in the variables $a_e$. We derive an explicit combinatorial formula for $\alpha_\Gamma$, and we prove that $\alpha_\Gamma \wedge \alpha_\Gamma=0$. 
}

\begin{document}

\maketitle

\tableofcontents

\section{Introduction}
Topological Quantum Field Theories are an important theoretical tool in physics and mathematics. A celebrated use of two-dimensional TQFT is Kontsevich's work on deformation quantization \cite{kontsevich_deformation_2003,cattaneo_path_2000}. In physical terms, that work proves two results about a certain \enquote{free B-model} 2d TQFT: 
\begin{itemize}
    \item The space of formal deformations of the theory receives no quantum corrections and coincides with the space of Poisson bivectors on the target space of the B-model.
    \item The formal deformations of the theory induce a deformation of the algebra of local operators at \enquote{Neumann} boundary conditions for the 2d TQFT. Quantum corrections produce the full deformation quantization of the algebra of functions on the target space of the B-model.
\end{itemize}
Both sets of quantum corrections are given by certain 2d Feynman integrals. The first point is proven by a careful algebraic-geometric analysis of the Feynman integrals in position space. Introducing Schwinger/Feynman parameters and integrating over positions, the remaining integrands take the form $\alpha_\Gamma \wedge \alpha_\Gamma$ for a 1-form $\alpha_\Gamma$ defined below and a \enquote{Laman} graph $\Gamma$ \cite{budzik_feynman_2023,gaiotto_higher_2024}.  

In this note we give an explicit construction of $\alpha_\Gamma$ (\cref{lem:dbarT})\footnote{A Mathematica \textregistered ~ implementation of this formula, and all examples in the present article, is attached to the online version and available from the first author's website \href{https://paulbalduf.com/research}{paulbalduf.com}.}, which then leads to a combinatorial proof of $\alpha_\Gamma \wedge \alpha_\Gamma=0$ for all graphs $\Gamma$ that are not trees (\cref{thm:QE_0}). In particular, this gives a direct combinatorial proof of the vanishing of the quantum corrections to the space of formal deformations of the free 2d B-model. Furthermore, the result for non-Laman graphs $\Gamma$ implies a similar result for certain higher-dimensional generalizations of the 2d B-model we now briefly review. 

An independent   proof of $\alpha_\Gamma \wedge \alpha_\Gamma=0$, based on 2-dimensional complex integration,  has been given in the recent article \cite{wang_factorization_2024}. Our work differs from that approach in that we construct explicit formulas for an individual factor $\alpha_\Gamma$, while the authors of \cite{wang_factorization_2024} only examine products like $(\alpha_\Gamma \wedge \alpha_\Gamma)$.

\subsection{Generalized B-models}
The 2d free B-model can be generalized to a collection of free QFTs which depends topologically on $T$ directions of a ${\mathbb R}^{T+2H}$ space-time and holomorphically on the remaining $2H$ directions. These theories and their deformations include 3d Chern-Simons theory, holomorphic-topological variants thereof with applications to integrability \cite{costello_supersymmetric_2013} and W-algebras \cite{costello_mtheory_2016}, as well as the  \enquote{twist}  of many super-symmetric QFTs of theoretical interest \cite{elliott_topological_2019}. The space of formal deformations of such QFTs is also of interest. The authors of 
\cite{gaiotto_higher_2024} concluded that quantum corrections should be absent whenever $T\geq 2$, by a comparison of the Feynman diagrams which appear in a QFT 
with a ${\mathbb R}^{T} \times {\mathbb C}^{H}$ splitting of ${\mathbb R}^{T+2H}$ and one with a ${\mathbb R}^{T-2} \times {\mathbb C}^{H+1}$ splitting. 

The authors of \cite{gaiotto_higher_2024} also defined a certain form $\alpha_\Gamma$ on the space of Feynman parameters for any graph $\Gamma$, which captures the contribution of a single topological direction to 
the problem. They conjectured the combinatorial identity 
\cref{thm:QE_0},
\begin{align}\label{theorem1_2}
\qquad \alpha_\Gamma \wedge \alpha_\Gamma=0 \qquad \text{for all graphs }\Gamma \text{ except for trees}.
\end{align}
The proof of this identity gives a simple combinatorial proof of the vanishing of quantum corrections for all these theories with $T\geq 2$.\footnote{The identity also implies the vanishing result in the presence of any number of defects in the theories, as long as they extend along the two topological directions and are built from free fields. }

\subsection[The definition of the integral alpha]{The definition of $\alpha_\Gamma$}
Let $\Gamma$ be a graph with edges $E_\Gamma$ and vertices $V_\Gamma$. There is only one type of edges, and there is no restriction on the valence of vertices. The graph is not necessarily connected, and it is not necessarily a simple graph (multi-edges are allowed), and it might or might not be a tree.

For each edge $e\in E_\Gamma$, introduce a Schwinger parameter $a_e$. We are interested in computing an integrand in the space of Schwinger parameters, so we treat $a_e$ it as a formal real variable, assuming $a_e\neq 0$. All edges have an arbitrary, but fixed, direction, $e=(v^-(e) \rightarrow  v^+(e))$ where $v^\pm(e)$  are vertices. 
For each vertex $v\in V_\Gamma$, introduce a position variable $x_v\in \mathbb R$. We let $x^+_e, x^-_e$ be the position variables of the end vertices $v^+(e), v^-(e)$ of an edge $e$.
Using $a_e$ and $x_v$, we define the edge variable 
\begin{align}\label{def:se}
s_e &:= \frac{x^+_e-x^-_e}{\sqrt{a_e}} \in \mathbb R.
\end{align}
The object of interest of the present article is the differential form introduced in  \cite{budzik_feynman_2023,gaiotto_higher_2024},
\begin{align}\label{def:omega1}
	\alpha_\Gamma := \idotsint \limits_{\mathbb R^{\abs{V_\Gamma}-1}}\bigwedge_{e\in E_\Gamma}e^{-s_e^2} \;  \d s_e.
\end{align}
$\alpha_\Gamma$ is a constant (0-form) for tree graphs, and it is a non-trival differential form in the Schwinger parameters $\left \lbrace a_e \right \rbrace   $ when $\Gamma$ has loops.
In order to prove \cref{thm:QE_0} (\cref{theorem1_2}), we proceed in three steps: Firstly, in \cref{sec:integrand}, we establish a formula for the integrand of $\alpha_\Gamma$ as a sum of spanning trees, taking into account all signs that can arise from (arbitrary) choices of labeling or edge directions in $\Gamma$. Secondly, in \cref{sec:integration}, we compute the integral in \cref{def:omega1} in terms of Dodgson polynomials. Thirdly, in \cref{sec:wedge_product}, we use combinatorial properties of Dodgson polynomials to show that $\alpha_\Gamma \wedge \alpha_\Gamma=0$ for all graphs except trees. 

Our constructions make heavy use of graph matrices, a reminder of their definitions is \cref{sec:graph_matrices}. In \cref{sec:Dodgson_polynomials}, we review the definitions of Dodgson polynomials and prove technical lemmas that are used in the main text. Some of our constructions are analogous to the derivation of parametric Feynman rules from momentum-space Feynman rules, reviewed in \cref{sec:parametric_feynman_integral}.

\acknowledgments
PHB thanks Karen Yeats for multiple discussions. This research was supported in part by a grant from the Krembil Foundation. DG
is supported by the NSERC Discovery Grant program and by the Perimeter Institute
for Theoretical Physics. Research at Perimeter Institute is supported in part by the Government of Canada through the Department of Innovation, Science and Economic Development and by the Province of Ontario through the Ministry of Colleges and Universities. 
 
\section{Setup}\label{sec:setup}
\subsection{Edge variables and vertex variables}
To  make \cref{def:omega1} amenable to a mathematical analysis, we   express  the edge variable of \cref{def:se} in terms of the graph matrices (\cref{sec:graph_matrices}). 
We collect all $\abs{V_\Gamma}$ vertex coordinates into one vector $\vec x$. The difference in the numerator of \cref{def:se} is then the $e$\textsuperscript{th} entry of the signed incidence matrix $\bar \incidencematrix$ (\cref{def:incidence_matrix}) applied to the vector $\vec x$, 
\begin{align}\label{se_vector}
s_e &= \frac{1}{\sqrt{a_e}} \left( \bar {\incidencematrix} \vec x \right) _e, \hspace{3cm} 	\vec s := \frac{1}{\sqrt{\vec a}} \bar {\incidencematrix}  \vec x.
\end{align}
 \Cref{se_vector} involves a component-wise product between the vectors $\bar{\incidencematrix}\vec x$ and $\vec a$, not a dot product. One can rewrite this as a matrix product if one introduces   a $\abs{E_\Gamma}\times \abs{E_\Gamma}$  diagonal matrix according to \cref{def:edge_variable_matrix}:
\begin{align*}
	\edgematrix  & := \operatorname{diag}(\vec a)= \operatorname{diag}\left(  a_1 ,  \ldots,  a_{\abs{E_\Gamma}} \right).
\end{align*}
This leads us to the fully \enquote{vectorized} form of \cref{def:se},
\begin{align}\label{se_vector_full}
\vec s &=  \edgematrix^{-\frac 12}  \bar {\incidencematrix} \vec x.
\end{align}

\begin{example}{Multiedge}\label{ex:multiedge_svector}
	The incidence matrix (\cref{def:incidence_matrix}) depends on the labeling and direction of edges and labeling of vertices. One possible choice for the 1-loop multiedge (\enquote{bubble} or \enquote{fish}) graph is 
	
	\begin{minipage}{.3\linewidth}
		\begin{tikzpicture} 
			\node at (-1.5,0){$\Gamma=$};
			\node[vertex, label=left:{$v_1$}](v1) at (0,0){};
			\node[vertex, label=right:{$v_2$}](v2) at (2.2,0){};
			
			\draw[edge, ->, bend angle=30,bend left](v1) to node[pos=.5, above]{$e_1$}(v2);
			\draw[edge, ->, bend angle=30, bend right](v1) to node[pos=.5, below]{$e_2$}(v2);
		\end{tikzpicture}
	\end{minipage}\hspace{1cm}
	\begin{minipage}{.6\linewidth}
	\begin{align*}
		\bar {\incidencematrix} &= \begin{pmatrix}
			-1 & 1 \\
			-1 & 1
		\end{pmatrix}, \qquad \mathbb  D = \begin{pmatrix}
			a_1 & 0 \\
			0 & a_2
		\end{pmatrix}.
	\end{align*}   
	\vspace{-.05cm}
	\end{minipage}
		
	\noindent
	The vector $\vec s$ from \cref{se_vector_full} has the form expected from \cref{def:se},
	\begin{align*}
	\vec s &= \begin{pmatrix}
		\frac 1 {\sqrt{a_1}} & 0 \\
		0 & \frac 1 {\sqrt{a_2}}
	\end{pmatrix}\begin{pmatrix}
	-1 & 1 \\
	-1 & 1
	\end{pmatrix}\begin{pmatrix}
	x_1\\ x_2
	\end{pmatrix} = \begin{pmatrix}
	\frac{x_2-x_1}{\sqrt{a_1}} \\ \frac{x_2-x_1}{\sqrt{a_2}}
	\end{pmatrix}.
	\end{align*}

\end{example}

In what follows, we will be integrating over the vertex positions $\d x_v$. Due to Lorentz invariance, we can set any one of these positions to zero instead of integrating. We choose the last index, $v_\star = v_{\abs{V_\Gamma}}$, hence we set $x_{\abs{V_\Gamma}}=0$. Effectively, this can be reached already at the level of \cref{se_vector_full} by leaving out the last column of the incidence matrix $\bar {\incidencematrix}$, turning it into the reduced incidence matrix $\incidencematrix$ (\cref{def:reduced_incidence_matrix}). Unless mentioned otherwise, we will from now on assume that the vector $\vec x$ does not contain $x_{\abs{V_\Gamma}}$, and that  
\begin{align}\label{def:se_vector}
	\vec s &=  \edgematrix^{-\frac 12}    \incidencematrix \vec x.
\end{align}

\begin{example}{Dunce's cap} \label{ex:dunces_cap_svector}
	The dunce's cap   is a graph on 3 vertices and 4 edges. 
		
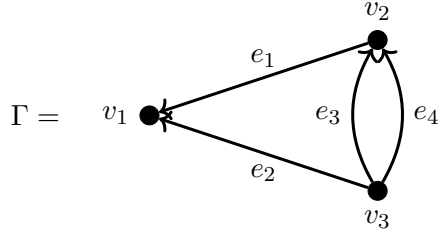
\begin{center}
		\begin{tikzpicture}
			\node at (-1.5,0){$\Gamma=$};
			\node[vertex, label=left:{$v_1$}](v1) at (0,0){};
			\node[vertex, label={$v_2$}](v2) at (3,1){};
			\node[vertex, label=below:{$v_3$}](v3) at (3,-1){};
			
			\draw[edge, ->](v2) to node[pos=.5, above]{$e_1$}   (v1) ;
			\draw[edge, ->](v3) to node[pos=.5, below]{$e_2$}(v1);
			\draw[edge, ->, bend angle=30,bend left](v3) to node[pos=.5, left]{$e_3$}(v2);
			\draw[edge, ->, bend angle=30, bend right](v3) to node[pos=.5, right]{$e_4$}(v2);
		\end{tikzpicture}
	\end{center}
	We choose $v_3=v_\star$ as the vertex to remove, then 
	\begin{align*}
	   \incidencematrix  &= \begin{pmatrix}
			1 & -1  \\
			1 & 0  \\
			0 & 1  \\
			0 & 1  
		\end{pmatrix}, \qquad \edgematrix = \begin{pmatrix}
			a_1 & 0 & 0 &0 \\
			0 & a_2 & 0 & 0 \\
			0 & 0 & a_3 & 0 \\
			0 & 0 & 0 & a_4
		\end{pmatrix}.
	\end{align*}
	Again, we obtain the expected quantities $s_e$ from \cref{def:se_vector},
	\begin{align*}
	 	\vec s  &=\begin{pmatrix}
		\frac 1 {\sqrt {a_1}} & 0 & 0 &0 \\
		0 & \frac 1 {\sqrt {a_2}} & 0 & 0 \\
		0 & 0 & \frac 1 {\sqrt {a_3}} & 0 \\
		0 & 0 & 0 &	\frac 1 {\sqrt {a_4}}
		\end{pmatrix}\begin{pmatrix}
		1 & -1  \\
		1 & 0  \\
		0 & 1  \\
		0 & 1  
		\end{pmatrix} \begin{pmatrix}
		x_1\\ x_2 
		\end{pmatrix}= \begin{pmatrix}
		\frac{x_1-x_2}{\sqrt{a_1}}\\ \frac{x_1  }{\sqrt{a_2}} \\
		\frac{x_2  }{\sqrt{a_3}} \\ \frac{x_2  }{\sqrt{a_4}}
		\end{pmatrix}.
	\end{align*} 
Alternatively, if we had used the non-reduced incidence matrix and \cref{se_vector_full}, we would have obtained functions involving $x_3$. Setting $x_3=0$ then reproduces the present result. 
	
\end{example}

The integrand in  \cref{def:omega1} involves the quantity $\prod_{e\in E_\Gamma} e^{-s_e^2} = e^{ -\sum_e s_e^2}$.
 The sum in the exponent is the norm squared of the vector $\vec s$ (\cref{def:se_vector}):
\begin{align}\label{exponential_laplacian}
\sum_{e\in E_\Gamma} s_e^2 &= \vec s \cdot \vec s = \left( \edgematrix^{-\frac 12} \incidencematrix \vec x \right) ^T \left( \edgematrix^{\frac 12} \incidencematrix \vec x \right)= \vec x^T \incidencematrix^T \edgematrix^{-1} \mathbb  I \vec x=\vec x^T \laplacian \vec x.
\end{align}
We have identified the reduced  labeled Laplacian $\laplacian$ (\cref{def:Laplacian}). Note that  $e^{-\vec x \laplacian \vec x}$ is familiar from the   ordinary scalar Feynman rules in parametric space, see \cref{sec:parametric_feynman_integral}.  
\begin{example}{Dunce's cap}\label{ex:dunces_cap_exponential_factor}
	The reduced  Laplacian (\cref{def:Laplacian}) of the dunce's cap (\cref{ex:dunces_cap_svector}) is
	\begin{align*}
	\laplacian &=   \incidencematrix^T \edgematrix^{-1}   \incidencematrix = \begin{pmatrix}
		\frac{1}{a_1} + \frac{1}{a_2} & -\frac{1}{a_1}  \\
		-\frac{1}{a_1} & \frac{1}{a_1} + \frac{1}{a_3} + \frac{1}{a_4}  
	\end{pmatrix}.
	\end{align*}
	The square of the $\vec s$-vector is, after sorting the terms, the sum over the $s_e^2$ for all edges $e$, as one should expect.
	\begin{align*}
	{\vec s~}^2=\vec x^T \laplacian \vec x = \frac{\left( x_1-x_2 \right) ^2}{a_1} + \frac{ x_1  ^2}{a_2} + \frac{   x_2  ^2}{a_3 } + \frac{   x_2  ^2}{a_4 }.
	\end{align*}
	
\end{example}

\subsection{Edge Differentials and integral}\label{sec:integral}
The differential form part $ds_e$ is the main difference between our calculation and more conventional Feynman integrals in parametric space, compare \cref{sec:parametric_feynman_integral}. We differentiate \cref{def:se}:
\begin{align}\label{ds}
\d s_e &= \d \left( a_e^{-\frac 12} \left( x_{v^+(e)}-x_{v^-(e)} \right)   \right) \nonumber  \\
&= -\frac 12 a_e^{-\frac 32} \left(x_{  v^+(e)}-x_{v^-(e)}  \right) \d a_e + a_e^{-\frac 1 2} \left( \d x_{v^+(e)}- \d x_{v^-(e)}  \right) \nonumber \\
&= -\frac 12 a_e^{-\frac 32} \left( \incidencematrix\vec x \right) _e \; \d a_e + a_e^{-\frac 12} \; \d  \left( \incidencematrix \vec x \right) _e.
\end{align}
Using 
\begin{align*}
	\d \mathbb  D &:= \operatorname{diag} \left( \d \vec a \right) = \operatorname{diag} \left( \d a_1, \ldots, \d a_{\abs{E_\Gamma}} \right) 
\end{align*}
and \cref{se_vector}, the differential of the entire vector $\vec s$ is
\begin{align}\label{ds_vector}
	\d \vec s	&= -\frac 12 \edgematrix^{-\frac 3 2} \; \d \edgematrix \; \incidencematrix \vec x + \edgematrix^{-\frac 12} \incidencematrix \; \d \vec x.
\end{align}
Recall that we have set $x_{\abs{V_\Gamma}}=0$, consequently, the vectors $\vec x$ and $\d \vec x$ have only $\left( \abs{V_\Gamma}-1 \right) $ components.

\begin{example}{Multiedge}\label{ex:multiedge_dsvector}
	For the multiedge of \cref{ex:multiedge_svector}, one finds
	\begin{align*}
	\d \vec s &= - \frac 12 \begin{pmatrix}
		a_1^{-\frac 3 2} & 0 \\
		0 & a_2^{-\frac 3 2}
	\end{pmatrix} \begin{pmatrix}
	\d a_1 & 0 \\
	0 & \d a_2
	\end{pmatrix}\begin{pmatrix}
	-1  \\
	-1 
	\end{pmatrix} \begin{pmatrix}
	x_1 
	\end{pmatrix} + \begin{pmatrix}
	a_1^{-\frac 1 2} & 0 \\
	0 & a_2^{-\frac 1 2}
	\end{pmatrix} \begin{pmatrix}
	-1   \\
	-1  
	\end{pmatrix}\begin{pmatrix}
	\d x_1  
	\end{pmatrix}\\
	&=    \begin{pmatrix}
		 \frac{  x_1   \d a_1}{2 a_1^{\frac 3 2}}\\
		 \frac{ x_1  \d a_2}{2 a_2^{\frac 3 2}}
	\end{pmatrix} + \begin{pmatrix}
	\frac{  - \d x_1}{a_1^{\frac 12}} \\\frac{ - \d x_1}{a_2^{\frac 12}}
	\end{pmatrix} = \begin{pmatrix}
	-\frac{1}{2 }a_1^{-\frac 32} \Big( - x_1 \;\d a_1 +2 a_1  \d x_1  \Big)  \\
	-\frac{1}{2 }a_2^{-\frac 32} \Big( - x_1 \;\d a_2 + 2 a_2   \d x_1  \Big)  
	\end{pmatrix}.
	\end{align*}
\end{example}

To remove trivial scalar factors from  $\d \vec s$  (\cref{ds_vector}), we introduce 1-forms $\vec \rho$ according to $-2 \d \vec s =  \edgematrix^{-\frac 3 2} \, \vec \rho$, or equivalently $\vec \rho:=  \d \edgematrix \; \incidencematrix \vec x -2 \edgematrix  \incidencematrix \; \d \vec x$.
For a single edge $e$, 
\begin{align}\label{def:edge_dp}
	\rho_e &= \left( \incidencematrix \vec x \right) _e \;\d a_e   - 2 a_e \left( \incidencematrix \; \d \vec x \right) _e.
\end{align}

\begin{example}{Dunce's cap}\label{ex:dunces_cap_dsvector}
	For the dunce's cap (\cref{ex:dunces_cap_svector}), one finds \cref{ds_vector}
	\begin{align*}
	\d \vec s &= -\frac 12 \begin{pmatrix}
		\frac{ (x_1-x_2)\d a_1}{  a_1^{\frac 3 2}} \\
		\frac{  x_1 \d a_2}{  a_2^{\frac 3 2}} \\
		\frac{  x_2 \d a_3}{  a_3^{\frac 3 2}} \\
		\frac{  x_2 \d a_4}{  a_4^{\frac 3 2}}
	\end{pmatrix} + \begin{pmatrix}
	\frac{\d x_1 - \d x_2}{a_1^{\frac 1 2}} \\
	\frac{    \d x_1}{a_2^{\frac 1 2}}\\
	\frac{   \d x_2}{a_3^{\frac 1 2}} \\
	\frac{    \d x_2}{a_4^{\frac 1 2}}
	\end{pmatrix}.
	\end{align*}
\Cref{def:edge_dp} then is
	\begin{align*}
	\rho_1 &= (x_1-x_2) \d a_1 - 2 a_1 \left( \d x_1 - \d x_2 \right) \\
	\rho_2 &= x_1 \d a_2 - 2 a_2 \d x_1   \\
	\rho_3 &=x_2 \d a_3 - 2 a_3  \d x_2   \\
	\rho_4 &=x_2 \d a_4 - 2 a_4  \d x_2 .
	\end{align*}
\end{example}
We have now at our disposal all ingredients to define  the integration measure   associated with the graph $\Gamma$ in \cref{def:omega}, $\bigwedge_{e\in E_\Gamma} e^{-s_e^2} \d s_e$.  It proves valuable to introduce the form 
\begin{align}\label{def:barP}
	\bar P_\Gamma &:= \rho_1 \wedge \rho_2 \wedge \ldots \wedge \rho_{\abs{E_\Gamma}} = 	\bigwedge_{e\in E_\Gamma} \rho_e.
\end{align}
 Expanding the product $\bar P_\Gamma$ gives a sum of terms, where each term is a choice of assigning to each edge of the graph one of the two summands in \cref{def:edge_dp}. In that sense, the two terms in $\rho_e$ (\cref{def:edge_dp}) can be interpreted as Feynman rules for two distinct types of edges, and   $\bar P_\Gamma$ is the sum over all possibilities to assign these types to the edges of $\Gamma$. 

\begin{example}{Multiedge}\label{ex:multiedge_U}
	For the multiedge, we have computed  $\d \vec s$ in \cref{ex:multiedge_dsvector},  and therefore
	\begin{align*}
	 & \bigwedge_{e\in E_\Gamma} e^{-s_e^2} \d s_e = e^{ - \frac{( -x_1)^2}{a_1}}  \d s_1 \wedge e^{- \frac{(-x_1)^2}{a_2}}\d s_2 = e^{ - \frac{( -x_1)^2}{a_1}- \frac{(-x_1)^2}{a_2}} \d s_1 \wedge \d s_2 \\
	&=  \frac {e^{ - \frac{( -x_1)^2}{a_1}- \frac{( -x_1)^2}{a_2}} } {4 a_1^{\frac 32} a_2 ^ {\frac 32}} \cdot  
	\Big(  -x_1  \d a_1 +2 a_1 \d x_1 \Big)    \wedge  \Big(  -x_1 \d a_2 +2 a_2 \d x_1 \Big)  \\
		&=  \underbrace{\frac {e^{ - x_1 \left( \frac{1}{a_1}+ \frac{1}{a_2} \right) x_1} } {4 a_1^{\frac 32} a_2 ^ {\frac 32}}}_{\text{scalar}} 
		\underbrace{   x_1^2  \left( \d a_1 \wedge \d a_2 \right)   + 2 x_1 a_1   \left(   \d a_2 \wedge \d x_1  \right)- 2 x_1a_2  \left(    \d a_1 \wedge \d x_1 \right)  }_{\bar P_\Gamma}. 
	\end{align*}
In the last line, we have identified the factor $\bar P_\Gamma $ according to \cref{def:barP}.
\end{example}

\begin{example}{Dunce's cap}\label{ex:dunces_cap_W}
	For the dunce's cap, starting from \cref{ex:dunces_cap_dsvector}, a tedious but straightforward calculation produces
	\begin{align*}
	\bar P_\Gamma &=x_2 \Big(  	4  a_1 a_4 x_1  \left( \d a_2 \wedge \d a_3 \wedge \d x_1 \wedge \d x_2   \right) ~-4  a_1 a_3 x_1\left( \d a_2 \wedge \d a_4\wedge \d x_1\wedge \d x_2   \right)  \\
	&\qquad  -	4 a_2 a_4 x_1 \left(\d a_1 \wedge \d a_3 \wedge \d x_1 \wedge \d x_2   \right) ~ -  4 a_2 a_3  x_2 \left(\d a_1 \wedge \d a_4 \wedge \d x_1 \wedge \d x_2     \right)  \\
	&\qquad   + 4 a_2 a_4 x_2 \left(\d a_1 \wedge \d a_3 \wedge \d x_1 \wedge \d x_2  \right) ~ +  4 a_2 a_3  x_1  \left(\d a_1 \wedge \d a_4 \wedge \d x_1 \wedge \d x_2     \right) \\
	&\qquad + 4 a_1 a_2 x_2  \left(\d a_3 \wedge \d a_4 \wedge \d x_1 \wedge \d x_2  \right)~ + 	2 a_3  x_1^2 \left(\d a_1 \wedge \d a_2 \wedge \d a_4 \wedge \d x_2   \right)     \\
	&\qquad - 2 a_4 x_1^2  \left(\d a_1 \wedge \d a_2 \wedge \d a_3 \wedge \d x_2   \right) ~ + 2 a_1 x_1 x_2  \left(\d a_2 \wedge \d a_3 \wedge \d a_4 \wedge \d x_1  \right)   \\
	&\qquad - 2 a_1 x_1 x_2  \left(\d a_2 \wedge \d a_3 \wedge \d a_4 \wedge \d x_2  \right)   ~ - 2 a_2 x_1 x_2 \left(\d a_1 \wedge \d a_3 \wedge \d a_4 \wedge \d x_1  \right)    \\
	&\qquad - 2 a_3 x_1x_2  \left(\d a_1 \wedge \d a_2 \wedge \d a_4 \wedge \d x_2    \right)  ~ + 2 a_4 x_1 x_2 \left(\d a_1 \wedge \d a_2 \wedge \d a_3 \wedge \d x_2    \right)  \\
	&\qquad+ 2 a_2 x_2^2  \left(\d a_1 \wedge \d a_3 \wedge \d a_4 \wedge \d x_1   \right)   ~- x_1 x_2^2	\left(\d a_1 \wedge \d a_2 \wedge \d a_3 \wedge \d a_4  \right)  \\
	&\qquad  + x_1^2 x_2\left(\d a_1 \wedge \d a_2 \wedge \d a_3 \wedge \d a_4   \right)   \Big).
	\end{align*}
	We note that each of the summands contains 4 differentials, reflecting the 4 edges of $\Gamma$. Some terms contain all 4 edge differentials $\d a_e$ while others contain up to two vertex differentials $\d x_v$. 
\end{example}

Knowing $\bar P_\Gamma$  from \cref{def:barP}, we now recall that $x_{\abs{V_\Gamma}}=0$, and we integrate the remaining   $\abs{V_\Gamma}-1$ vertex coordinates $x_v$ over all $\mathbb R$   according to \cref{def:omega1}. The resulting expression,
\begin{align}\label{def:omega}
	\alpha_\Gamma &=  \frac{1}{ (-2)^{\abs{E_\Gamma}}  \prod_e a_e^{\frac 3 2} }  \idotsint \limits_{\mathbb R^{\abs{V_\Gamma}-1}}e^{-\vec x^T \laplacian \vec x}\; \bar P_\Gamma,
\end{align}
is a function of $\abs{E_\Gamma}$ Schwinger parameters $a_e$ and their differentials $\d a_e$, that is, it is a differential form itself.

\begin{example}{Multiedge}\label{ex:multiedge_omega}
	For the multiedge (\cref{ex:multiedge_U}), we integrate over only a single variable, $x_1$. Explicitly, the integral is
	\begin{align*}
	\alpha_\Gamma  &= \frac{1}{ 2^{\abs{E_\Gamma}}  \prod_e a_e^{\frac 3 2} } \int \limits_{-\infty}^\infty  e^{-\vec x ^T\laplacian \vec x} \; \bar P_\Gamma
	= \int \limits_{-\infty}^\infty  \frac {e^{ - x_1 \left( \frac{1}{a_1}+ \frac{1}{a_2} \right) x_1} } {4 a_1^{\frac 32} a_2 ^ {\frac 32}}  x_1   x_1   \left( \d a_1 \wedge \d a_2 \right)  \\
	 &\quad +\int \limits_{-\infty}^\infty  \frac {e^{ - x_1 \left( \frac{1}{a_1}+ \frac{1}{a_2} \right) x_1} } {4 a_1^{\frac 32} a_2 ^ {\frac 32}}  x_1 2 a_1   \left(   \d a_2 \wedge \d x_1  \right) -\int \limits_{-\infty}^\infty  \frac {e^{ - x_1 \left( \frac{1}{a_1}+ \frac{1}{a_2} \right) x_1} } {4 a_1^{\frac 32} a_2 ^ {\frac 32}}  x_1 2 a_2  \left(    \d a_1 \wedge \d x_1 \right).  
	\end{align*}
	The first integral vanishes because it is supposed to be over $x_1$, but the integrand does not contain $\d x_1$. The remaining two integrals are
	\begin{align*}
	 \alpha_\Gamma &= \frac{   a_1}{2 a_1^{\frac 32} a_2 ^ {\frac 32}}  \d a_2 \wedge  \int \limits_{-\infty}^\infty   e^{ - x_1 \left( \frac{1}{a_1}+ \frac{1}{a_2} \right) x_1}    x_1  \;  \d x_1  \\
	 &\qquad -  \frac{   a_2}{2 a_1^{\frac 32} a_2 ^ {\frac 32}}  \d a_1 \wedge  \int \limits_{-\infty}^\infty  e^{ - x_1 \left( \frac{1}{a_1}+ \frac{1}{a_2} \right) x_1}   x_1   \; \d x_1    \\
	 &= 0. 
	\end{align*}
	Each of the two remaining integrals vanishes because it is an antisymmetric integrand proportional to $ x_1 \d x_1$, integrated over a symmetric domain. 
\end{example}

\section{Integrand in terms of trees}\label{sec:integrand}

The polynomial $\bar P_\Gamma$ from \cref{def:barP}, 
\begin{align}\label{def:barP2}
	\bar P_\Gamma &:= \rho_1 \wedge \rho_2 \wedge \ldots \wedge \rho_{\abs{E_\Gamma}},
\end{align}
contains  many summands, most of which vanish under the integral \cref{def:omega}. Only such terms survive which have exactly $(\abs{V_\Gamma}-1)$ factors of $\d x_j$. 
Clearly, $\bar P_\Gamma$ has no summands with more than $(\abs{V_\Gamma}-1)$ differentials $\d x_i$ because this is the number of vertices, and due to the antisymmetry of the wedge product  no vertex can appear twice. But there will be summands in $\bar P_\Gamma$ with less than $ (\abs{V_\Gamma}-1)$ vertex differentials. 

Knowing that the latter vanish under the integral, we will from now on leave out these terms and define 
\begin{align}\label{def:P}
P_\Gamma &:= \bar P_\Gamma \Big|_{\text{only terms containing $(\abs{V_\Gamma}-1)$ vertex differentials $\d x_j$}}.
\end{align}
On the other hand, we can use  \cref{def:edge_dp} for $\rho_e$ in \cref{def:barP2} and obtain
\begin{align}\label{Pgamma_factors}
\bar P_\Gamma &= \big( \d a_1 (\incidencematrix\vec x)_1 - 2 a_1 (\incidencematrix\d \vec x)_1 \big) \wedge \ldots \wedge  \big( \d a_{\abs{E_\Gamma}} (\incidencematrix\vec x)_{\abs{E_\Gamma}}  - 2 a_{\abs{E_\Gamma}}  (\incidencematrix\d \vec x)_{\abs{E_\Gamma}}  \big).
\end{align}
In this form, we see that   exactly  $ (\abs{V_\Gamma}-1)$  factors contribute their $\d x$ part, and the remaining 
\begin{align}\label{loop_number}
	\abs{E_\Gamma}- \left( \abs{V_\Gamma}-1 \right) = L_\Gamma
\end{align}
factors contribute their $d a_e (\incidencematrix \vec x)_e$. This integer $L_\Gamma\geq 0$ is the \emph{loop number}\footnote{$L_\Gamma$ is the dimension of the cycle space of $\Gamma$, known as \emph{first Betti number} in the mathematical literature.} of $\Gamma$ 

Let $T$ be a set of exactly $(\abs{V_\Gamma}-1)$ edges $T \subset E_\Gamma$. The   terms in $P_\Gamma$ which have a chance of not vanishing under the integral \cref{def:omega} are then a subset of the sum over all such sets $E$, namely the edges in $T$ should contribute their $\d x$ term,
\begin{align}\label{Pgamma_factors2}
P_\Gamma &= \sum_{\substack{T \subseteq E_\Gamma \\
\abs{T}=\abs{V_\Gamma}-1}} \sgn(T) \left( \prod_{e \notin T} \left( \incidencematrix \vec x \right)_e \right) \left(  \prod_{e \in T}- 2 a_e  \right) \left(  \bigwedge_{e \notin T} \d a_e \right)     \left(  \bigwedge_{v \in V_\Gamma}   \d  x_v \right)  . 
\end{align}
Here, $\sgn(T)$ is a sign which will be discussed in \cref{sec:signs}. 
A special case of graph $\Gamma$ is a tree (that is, a connected graph without cycles, or with loop number $L_\Gamma=0$). A tree has exactly $\abs{V_\Gamma}-1 $ edges. Consequently, for a tree every edge contributes its $\d x$ factors and there are no terms proportional to $\d a_e$ or to $x_v$. 
If $\Gamma$ is a tree then there is only one choice  of selecting the edges $T$ in \cref{Pgamma_factors2}, namely   $T=E_\Gamma$, and no differentials $\d a_e$ appear at all:
\begin{align}\label{integrand_tree}
		P_\Gamma &=   \sgn(E_\Gamma)  \left(  \prod_{e \in E_\Gamma} -2 a_e  \right)   \left(  \bigwedge_{v \in V_\Gamma} \d \vec x _v \right)  \qquad \text{(if  $\Gamma$ is a tree)}. 
\end{align}

\begin{example}{Dunce's cap}\label{ex:W_computed}
	We gave $\bar P_\Gamma$ for the dunce's cap in \cref{ex:dunces_cap_W}. We now leave out all summands which don't contain both $\d x_1$ and $\d x_2$.  The result is 
	\begin{align*}
		&P_\Gamma= 4x_2 x_1  \Big( -   a_1 a_3  \left( \d a_2 \wedge \d a_4  \right)   + 	   a_1 a_4  \left( \d a_2 \wedge \d a_3    \right)   \\
		&\qquad \qquad   -  a_2 a_4   \left(\d a_1 \wedge \d a_3   \right)  +   a_2 a_3  \left(\d a_1 \wedge \d a_4     \right)\Big)\wedge \d x_1\wedge \d x_2 \\
		&\quad  +4 x_2^2 \Big(   -  a_2 a_3  \left(\d a_1 \wedge \d a_4     \right)    +    a_1 a_2    \left(\d a_3 \wedge \d a_4   \right)   +   a_2 a_4   \left(\d a_1 \wedge \d a_3  \right)      \Big)\wedge \d x_1\wedge \d x_2.
	\end{align*}
	One could be tempted to assume that the first summand vanishes upon integration due to being linear in $x_1$ and $x_2$, similar to \cref{ex:multiedge_omega}. But this is not so. The reason is that the Laplacian $\laplacian$   in  $e^{-\vec x^T \laplacian \vec x}$ is not diagonal. Changing the variables of integration such that the matrix becomes diagonal, the factor $x_1 x_2$ transforms into non-vanishing terms like $x_1^2, x_2^2$. 
\end{example}

We do not yet know in general how many summands there  are in  \cref{Pgamma_factors2} because each of the factors $(\incidencematrix \d \vec x)_e$ and  $(\incidencematrix\vec x)_e$ can consist of one or two terms. We now examine how these factors interact. 

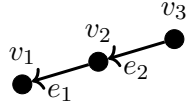
\begin{figure}[htb]
	\begin{center}
		\begin{tikzpicture}
			\node[vertex, label=above:{$v_1$}](v1) at (0,0){};
			\node[vertex, label=above:{$v_2$}](v2) at (1,.3){};
			\node[vertex, label=above:{$v_3$}](v3) at (2,.6){};
			
			\draw[edge, ->](v2) to node[pos=.5, below]{$e_1$}   (v1) ;
			\draw[edge, ->](v3) to node[pos=.5, below]{$e_2$}(v2); 
		\end{tikzpicture}
	\end{center}
	\caption{Two edges $e_1, e_2$   adjacent to the same vertex $v_2$.}
	\label{fig:edges_shared_vertex}
\end{figure}

For a given edge $e$, the factor $\left( \incidencematrix \vec x \right) _e$ is  the difference of the two vertices adjacent to $e$. Assume that two edges $e_1=(v_2 \rightarrow v_1)$ and $e_2=(v_3 \rightarrow v_2)$ share one vertex as in \cref{fig:edges_shared_vertex}. Then 
\begin{align*}
\left( \incidencematrix \vec x \right) _1 \left( \incidencematrix \vec x \right) _2 &= \left( x_1- x_2 \right) \left( x_2-x_3 \right) =x_1 x_2 - x_2^2 - x_1 x_3 + x_2 x_3.
\end{align*}
For the differentials, the quadratic term vanishes:
\begin{align*}
\left( \incidencematrix \d \vec x \right)_1 \wedge \left( \incidencematrix \d \vec x \right) _2 &= \left( \d x_1- \d x_2 \right) \wedge \left( \d x_2- \d x_3 \right)\\
&= \d x_1 \wedge \d x_2  - \d x_1 \wedge \d x_3 + \d x_2 \wedge \d x_3.
\end{align*}
This phenomenon can be interpreted graphically.  We have the two edges adjacent to the vertex $v_2$, and the non-vanishing terms in the product of differentials are exactly the choices of two distinct vertices. The possible choices are $\left \lbrace v_1,v_2 \right \rbrace $ or $\left \lbrace v_1,v_3 \right \rbrace $ or $\left \lbrace v_2,v_3 \right \rbrace $. Equivalently: For every edge $e$, select exactly one of the two adjacent vertices, and don't select any vertex twice.

On the other hand, for the integral \cref{def:omega} to be non-zero, every vertex except for the special vertex $v_\star$ must be selected. Hence, for every vertex except $v_\star$, we must select one of the adjacent edges. These pairs $(v,e)$ of a vertex and one of its adjacent edges are \enquote{flags}, or equivalently, half-edges. In a set $T$ of the sum \cref{Pgamma_factors2}, every edge appears at most once, hence, none of the pairs $(v,e)$ can include the same edge more than once.  

\begin{example}{Dunce's cap}\label{ex:dunces_cap_flags}
	Recall the dunce's cap from \cref{ex:dunces_cap_svector},
	\begin{center}
	\begin{tikzpicture}[scale=.7]
		\node at (-1.5,0){$\Gamma=$};
		\node[vertex, label=left:{$v_1$}](v1) at (0,0){};
		\node[vertex, label={$v_2$}](v2) at (3,1){};
		\node[vertex, label=below:{$v_3$}](v3) at (3,-1){};
		
		\draw[edge, ->](v2) to node[pos=.5, above]{$e_1$}   (v1) ;
		\draw[edge, ->](v3) to node[pos=.5, below]{$e_2$}(v1);
		\draw[edge, ->, bend angle=30,bend left](v3) to node[pos=.5, left]{$e_3$}(v2);
		\draw[edge, ->, bend angle=30, bend right](v3) to node[pos=.5, right]{$e_4$}(v2);
	\end{tikzpicture}
\end{center}

	Fixing $x_3=x_\star=0$, the two remaining vertices   allow for five different selections of flags:
	\begin{align*}
	&\left \lbrace (v_1, e_1), (v_2, e_3) \right \rbrace ,  \left \lbrace (v_1, e_1), (v_2, e_4) \right \rbrace\\
	&\left \lbrace (v_1, e_2), (v_2, e_1) \right \rbrace, \left \lbrace (v_1, e_2), (v_2, e_3) \right \rbrace, \left \lbrace (v_1, e_2), (v_2, e_4) \right \rbrace.
	\end{align*}
\end{example}

In \cref{ex:dunces_cap_flags}, the five possible ways to select $T$ correspond to the five spanning trees of $\Gamma$, and selecting the set of edges $T$ already implies the choice of flags, but that is coincidence because the graph is so small. 
In general, a fixed set $T$ of edges in \cref{Pgamma_factors2} may allow for several choices of flags. In particular, it might be possible to select a cycle of edges which does not include the special vertex $v_\star$ (it is impossible to select a cycle containing $v_\star$ because there is no $\d x_\star$, and hence no pair $(v_\star,e)$ for any edge $e$).

\begin{lemma}\label{lem:cycles_vanish}
	Fix a set of edges $T \subseteq E_\Gamma$ in the sum in \cref{Pgamma_factors2} for $P_\Gamma$. 
	When selecting flags, it is permissible to select flags such that they form a cycle, but the two possible directions of flags in a cycle will always cancel in the integrand. Consequently, one can leave out all flags that form a cycle.
\end{lemma}
\begin{proof}
    Let  $C:=\left \lbrace e_1, \ldots, e_n \right \rbrace \subset T$ be a $n$-edge cycle in the set $T$ of selected edges, where  $\left \lbrace v_a, v_b, \ldots, v_f \right \rbrace   $ are the vertices in the cycle (these vertices might be adjacent to additional edges not in the cycle). There are two ways of selecting the flags in the cycle, namely, the two directions of the cycle. 
	The first direction of selecting flags produces an expression proportional to 
	\begin{align*}
		\d x_a \wedge \d x_b \wedge \ldots \wedge \d x_e \wedge \d x_f,
	\end{align*}
	where there are $n$ factors in the wedge product. The second choice of direction means that for every edge, the other one of the two end vertices is chosen. Since the signs of vertices in an edge are opposite, this gives a $(-1)$ for each edge. Additionally, the order of differentials is altered because if we select the \emph{other} vertex at every edge, then, keeping the order of edges the same, all vertex indices get shifted by one. We obtain
	\begin{align*}
		\left( -\d x_f \right) \wedge \left( -\d x_a \right) \wedge \ldots \wedge (-\d x_e) .
	\end{align*}
	Let the number of edges $n$ in the cycle be even, then the number of vertices is even, too and there is an even number of minus signs. They cancel, leaving us with
	\begin{align*}
		\d x_f \wedge \d x_a \wedge \d x_b \wedge \ldots 
		&= (-1)^{n-1} \d x_a \wedge \d x_b \wedge \ldots \wedge \d x_b\\
		&= -\d x_a \wedge \d x_b \wedge \ldots \wedge \d x_b.
	\end{align*}
	This is the same term as for the first choice of flag direction, but with opposite sign. As the remaining parts of the integrand are unaffected by this choice, the two flags with opposite directions cancel out each other. 
	
	If the number $n$  of edges in the cycle is odd, the same mechanism is at work, but this time, an overall factor $(-1)$ remains from changing sign of all differentials. Conversely, one now needs an even number of permutations to restore the original order, and again, the two choices of cycle direction contribute the same term, but with opposite signs. 
	
\end{proof}

\begin{lemma}\label{lem:spanning_trees}
	\begin{enumerate}
		\item If $\Gamma$ is an arbitrary graph, then only those sets of flags contribute where their edges $T$ in \cref{Pgamma_factors2} form a spanning tree of $\Gamma$. 
		\item For every such spanning tree $T$, there is exactly one choice of flags that contributes to \cref{Pgamma_factors2}. That is, selecting the edges $T$ is equivalent to fixing the flags. 
	\end{enumerate}
\end{lemma}
\begin{proof}
	1.  We know that the relevant flags are those without cycles, hence they are forests. For a graph on $\abs{V_\Gamma}$ vertices, these flags contain $\abs{V_\Gamma}-1$ edges because every edge contributes the $\d x_v$ of one vertex, and they must cover every vertex except $v_\star$. A forest with $\abs{V_\Gamma}-1$ edges is necessarily adjacent to at least $\abs{V_\Gamma}$ distinct vertices. But the graph does not have more vertices than that, hence, the forest is adjacent to exactly $\abs{V_\Gamma}$ vertices. A forest with $\abs{V_\Gamma}-1$ edges and $\abs{V_\Gamma}$ vertices  is a tree (and not a disconnected forest). Hence, the edges $T$ of the flag form a tree that is adjacent to every vertex of $\Gamma$, that is, a spanning tree.
	
	2. Every vertex of $\Gamma$, even the special one $v_\star$, is incident to at least one edge of the flag. On the other hand, the coordinate $x_\star$ is not being integrated, hence, there is no $\d x_\star$. All edges of the tree which is adjacent to $v_\star$ must therefore contribute their \emph{other} vertex differential, say $\d x_i$, not $\d x_\star$. Hence, for these edges, there is only one choice to select the flag. But now, the differential of $x_i$ has been used, and none of the other edges of $T$ which are adjacent to $v_i$ can contribute $\d x_i$. Consequently, each of these edges has only one possible choice of flag. This argument continues until the flags of all edges are fixed. 
	
	Equivalently, one can interpret a choice of flag as a choice of direction for every edge. The special vertex $v_\star$ fixes the direction of all edges adjacent to it, say, oriented towards $v_\star$. Recursively, this fixes the directions of all edges in the tree, and $T$ is a \emph{directed} spanning tree where all edges are oriented towards $v_\star$.
\end{proof}

We  can now write $P_\Gamma$ (\cref{def:P}) as a sum over spanning trees $T$.   Rewriting \cref{Pgamma_factors2}, we have
\begin{align}\label{Pgamma_factors3}
	P_\Gamma &=  \sum_{T \text{ spanning}} \sgn(T)  \underbrace{\left( \prod_{e \notin T} \left( \incidencematrix \vec x \right)_e \right) }_{=:X_T} \left(  \prod_{e \in T} -2 a_e  \right) \left(  \bigwedge_{e \notin T} \d a_e \right)      \left(  \bigwedge_{v \in V_\Gamma}   \d  x_v \right) .
\end{align}
In \cref{Pgamma_factors3}, there are exactly $L_\Gamma$ (\cref{loop_number}) 
edge differentials $\d a_e$.
At the same time,  $P_\Gamma$ is a polynomial in $\left \lbrace x_i \right \rbrace $ of degree  $L_\Gamma$, because every edge not in the spanning tree produces a linear term $(\incidencematrix \vec x)_e$, and they form the polynomial $X_T$. 
The sum over the edges not in a spanning tree appears in the definition of the first Symanzik polynomial $\psi_\Gamma$ (\cref{def:Symanzik_polynomial}), so $P_\Gamma$ can be interpreted as a non-commuting version of $\psi_\Gamma$.

\begin{example}{Dunce's cap}\label{ex:dunces_cap_trees}
	For the dunce's cap with our choice of fixed vertex, we saw in \cref{ex:dunces_cap_flags} that it is not even possible to select a flag which  contains a cycle, but that is coincidence for this small example.  Using $x_3=0$, the sum \cref{Pgamma_factors3} contains five summands, which are
	\begin{align*}
		&  4 a_1 a_3     \left( x_1 - 0 \right) \left( x_2-0 \right) \left( \d a_2 \wedge \d a_4  \right)\left( \d x_1 \wedge \d x_2 \right), \\
		& 4 a_1 a_4  \left( x_1-0 \right) \left( x_2-0 \right) \left( \d a_2 \wedge \d a_3 \right) \left( \d x_1 \wedge \d x_2 \right), \\
		& 4 a_1 a_2  \left( x_2-0 \right) \left( x_2-0 \right)  \left( \d a_3 \wedge \d a_4 \right) \left( \d x_1 \wedge \d x_2 \right), \\
		&  4 a_2 a_3 \left( x_1-x_2 \right) \left( x_2-0 \right) \left( \d a_1 \wedge \d a_4 \right)  \left( \d x_1 \wedge \d x_2 \right),  \\
		& 4 a_2 a_4 \left( x_1-x_2 \right) \left( x_2-0\right) \left( \d  a_1 \wedge \d a_3 \right) \left( \d x_1 \wedge \d x_2 \right) .
	\end{align*}
	This reproduces the terms of $P_\Gamma$ from \cref{ex:W_computed}. Their signs will be discussed in \cref{sec:signs}.
	
\end{example}

\subsection{Signs}\label{sec:signs}

We now turn to the signs in  \cref{Pgamma_factors3}. Note that in each summand, the factors $\d a_e$ in the wedge product   are sorted in increasing order by  \cref{def:barP}.
There are three contributions to the sign $\sgn(T)$:
\begin{enumerate}
	\item Permuting the wedge products such that all $\d a_e$ come before all $\d x_v$.
	\item The $\d x_v$ come in an order that is implied by the labeling of edges, namely $(\incidencematrix \d \vec x)_e$. They need to be brought into increasing order.
	\item The $\d x_v$ themselves potentially carry a minus sign, depending on whether they are the start or the end of an edge. 
\end{enumerate}
Note that the latter two points only concern those edges which are in the tree $T$. Conversely, the edges not in the tree make the factor $X_T$ in \cref{Pgamma_factors3}, which also carries some sign, but that sign is  not included in  $\sgn(T)$.

\subsubsection[Permuting da out of dx]{Permuting $\d a$ out of $\d x$}

The first point in the above list is relatively straightforward. We know that $\d a_j$ appears at position $j$ in the original wedge product \cref{def:barP}, so it acquires a sign $(-1)^{k-j}$ when it is moved into position $k$. The $\d a_j$ were initially ordered, and we want them to stay ordered, so they never cross each other. The only thing that matters is how many $\d x_v$ they need to pass when they are moved in front of all the $\d x$'s. 

In the simplest case, the tree $T$ consists of exactly the last $(\abs{E_\Gamma}-L_\Gamma)$ edges. Then the first $L_\Gamma$ edges are not in the tree, and the first $L_\Gamma$ factors in the wedge product are $\d a_j$ and they never need to cross any $\d x$. These $\d a_j$ have indices $\left \lbrace 1,2,\ldots, L_\Gamma \right \rbrace $, and the sum of their indices $j$   is $\frac{L_\Gamma(L_\Gamma+1)}{2}$.

Every move of a single element $\d a_j$ by one step flips the sign. Simultaneously, when an element of $T$ is moved by one, the sum of the indices occupied by $T$ changes by one. We want the $\d a$'s to be in front of all $\d x$'s, hence, for an arbitrary $T$, the number of permutation steps is given by the sum of indices in $T$. Overall, the sign produced by moving the $\d a$'s in front of the $\d x$'s is
\begin{align*}
\left( -1 \right) ^{\sum_{e_j\notin T} j - \frac{L_\Gamma(L_\Gamma+1)}{2} }.
\end{align*}
If $L_\Gamma$ is even, then $\frac {L_\Gamma} 2$ is an integer and hence $\frac{L_\Gamma^2}{4}$ is an integer and therefore $\frac{L^2_\Gamma}{2}$ is even. An even number does not change the sign, so  we may write 
\begin{align}\label{sign_factor1}
	\left( -1 \right) ^{\sum_{e_j\notin T} j - \frac {L_\Gamma}2 } \qquad \text{for even $L_\Gamma$}.
\end{align}
Having done these permutations, all $\d a$'s are in front of all $\d x$'s, and the $\d a$'s are still ordered increasingly. But the $\d x$'s may be in any order, as implied by the incidence matrix.

\subsubsection{Signs arising from the incidence matrix}

The remaining sign arises from two effects that are closely related: The relative ordering of the $\d x_j$, and the signs they obtain from being at the beginning or the end of an edge. Both are directly given by the incidence matrix $\incidencematrix$ (\cref{def:incidence_matrix}) since we are considering $\prod_{e \in T}(\incidencematrix \d \vec x)_e$. We want to bring these $\d x$ into   increasing order. By $\incidencematrix[T]$, we denote the incidence matrix $\incidencematrix$ restricted to the rows that correspond to edges in $T$, compare \cref{def:minors}. Now observe that if $\incidencematrix[T]$ had no zeroes on the diagonal, the factors $\d x_v$ would be sorted increasingly and the sign would be given by the product of the diagonal entries of $\incidencematrix[T]$ (there is no contribution by off-diagonal elements because the requirement that $T$ is a tree exactly means that there is only one choice of edges in $T$, namely $T$ itself, that connects all vertices).

If $\incidencematrix[T]$ has zeroes on the diagonal, we obtain an additional sign from the permutation of rows that is required to produce a main diagonal without zero entries. But this sign is exactly what a determinant computes: $\det(\incidencematrix[T])$ equals the product of diagonal entries for a diagonal matrix, and otherwise the sign of the required row permutation times the product of the would-be diagonal entries. Hence, the sign generated from permuting the $\d x_v$, and from their individual signs, is given by $\det \left( \incidencematrix[T] \right) $. Another way to see this is to start from the axiomatic definition of the determinant: $\det(\incidencematrix[T])$ will flip its sign when two rows are exchanged, or when any row in $\incidencematrix[T]$ is multiplied by $-1$. The first case amounts to exchanging two factors in the wedge product of the $\d x_v$, the second case amounts to flipping the direction of any one edge in $T$.

Combining this result with \cref{sign_factor1}, we have established: 
\begin{lemma}\label{lem:sign_T}
	Let $\Gamma$ be a graph of even loop order $L_\Gamma$. 
	Let $T\subseteq E_\Gamma$ be a spanning tree (where the edges are ordered increasingly with respect to their label). Let $\incidencematrix[T]$ be the reduced incidence matrix, restricted to the edges in $T$. Then the sign of the differential $\bigwedge \d a_{\Gamma\setminus T}$ in \cref{Pgamma_factors3} is 
	\begin{align*} 
		\sgn(T) &=   (-1)^{ \sum_{e_i\notin T}i-\frac {L_\Gamma} 2} \det\big( \incidencematrix[T] \big) .
	\end{align*}
\end{lemma}

\begin{example}{Dunce's cap}\label{ex:dunces_cap_signs}
	From \cref{ex:dunces_cap_svector} we have for the dunce's cap
		\begin{align*}
		\incidencematrix &= \begin{pmatrix}
			1 & -1  \\
			1 & 0 \\
			0 & 1 \\
			0 & 1 
		\end{pmatrix}.
	\end{align*}
	The graph has five distinct spanning trees $T$, each of which consists of two edges. 
Plugging this into \cref{lem:sign_T}, we find: 
	\begin{align*}
		T   &: \hspace{.7cm}\left \lbrace 1,3 \right \rbrace \hspace{1.25cm}   \left \lbrace 1,4 \right \rbrace \hspace{1.25cm}\left \lbrace 1,2 \right \rbrace \hspace{1.2cm} \left \lbrace 2,3 \right \rbrace \hspace{1.25cm}\left \lbrace 2,4 \right \rbrace\\
			\incidencematrix[T]&: \quad  \begin{pmatrix}
			1 & -1\\
			0 & 1
		\end{pmatrix}  \hspace{.6cm} \begin{pmatrix}
			1 & -1\\
			0 & 1
		\end{pmatrix}  \hspace{.6cm} \begin{pmatrix}
			1 & -1\\
			1 & 0
		\end{pmatrix}  \hspace{.7cm} \begin{pmatrix}
			1 & 0\\
			0 & 1
		\end{pmatrix}   \hspace{.9cm} 		\begin{pmatrix}
			1 &0\\
			0 & 1
		\end{pmatrix}\\
		\det\big(\incidencematrix[T]\big)&: \hspace{1.1cm} 1 \hspace{2.1cm} 1 \hspace{2.1cm} 1 \hspace{2cm} 1 \hspace{2.1cm} 1\\
		(-1)^{ \sum_{e_i\notin T}i-\frac {L_\Gamma} 2}&:  \quad  (-1)^{2+4-1} \hspace{.5cm} (-1)^{2+3-1} \hspace{.5cm} (-1)^{3+4-1} \hspace{.5cm} (-1)^{1+4-1} \hspace{.5cm} (-1)^{1+3-1}\\
	\sgn(T)&:\hspace{.9cm}   -1  \hspace{1.6cm}  +1  \hspace{1.6cm} +1  \hspace{1.6cm} +1  \hspace{1.6cm} -1
	\end{align*}
    It is coincidence that for our choice of labels and directions, $\det(\incidencematrix[T])=+1$ for all spanning trees. 
	Writing the   summands from \cref{ex:dunces_cap_trees} together with the so-obtained signs, we find
	\begin{align*}
		P_\Gamma &= (-1)\cdot  4 a_1 a_3     x_1  x_2 \cdot      \left( \d a_2 \wedge \d a_4  \right)\left( \d x_1 \wedge \d x_2 \right) \\
		&\quad +  (+1)\cdot4 a_1 a_4   x_1    x_2 \cdot   \left( \d a_2 \wedge \d a_3 \right) \left( \d x_1 \wedge \d x_2 \right) \\
		&\quad + (+1) \cdot  4 a_1 a_2   x_2^2   \cdot \left( \d a_3 \wedge \d a_4 \right) \left( \d x_1 \wedge \d x_2 \right) \\
		&\quad +  (+1) \cdot 4 a_2 a_3 \left( x_1-x_2 \right) x_2  \cdot  \left( \d a_1 \wedge \d a_4 \right)  \left( \d x_1 \wedge \d x_2 \right)  \\
		&\quad + (-1)\cdot  4 a_2 a_4 \left( x_1-x_2 \right)  x_2   \cdot  \left( \d  a_1 \wedge \d a_3 \right) \left( \d x_1 \wedge \d x_2 \right) .
	\end{align*}
Upon expanding the parentheses, this result  coincides with the brute-force computation of \cref{ex:W_computed}, as expected.
\end{example}

\section{Integration}\label{sec:integration}

By \cref{lem:sign_T}, the sign of the contribution of a tree $T$ to $P_\Gamma$ is given by the determinant of $\incidencematrix[T]$, which vanishes when $T$ is not a tree. So, we may rewrite \cref{Pgamma_factors3} as a sum over all sets of $(\abs{V_\Gamma}-1)$ edges, and have
\begin{align}\label{Pgamma_factors4}
	P_\Gamma &= \underbrace{\sum_{\abs{T}=\abs{V_\Gamma}-1} (-1)^{  \sum_{e_i\notin T}i-\frac L 2} \det\big(\incidencematrix[T]\big)   \left(  \prod_{e \in T} -2 a_e  \right)\left( \prod_{e \notin T} \left( \incidencematrix \vec x \right)_e \right)  \left(  \bigwedge_{e \notin T} \d a_e \right)  }_{=:  W_\Gamma(\vec x)}     \bigwedge_{v\in V_\Gamma\setminus \left \lbrace v_\star \right \rbrace   } \d x_v   .
\end{align}
In this sum, all non-commuting factors are understood to be in their unique increasing order. From now on, we will leave out the index $_\Gamma$ from the loop order $L_\Gamma$ (\cref{loop_number}) as long as it refers to the same graph.

The integral we want to solve is given by \cref{def:omega,def:P},
\begin{align}\label{omega_integral} 
	\alpha_\Gamma &:=  \frac{1}{ (-2)^{\abs{E_\Gamma}}  \prod_e a_e^{\frac 3 2} }  \idotsint \limits_{\mathbb R^{\abs{V_\Gamma}-1}}e^{-\vec x^T \laplacian \vec x}\;   W_\Gamma(\vec x) \bigwedge_{v \in V_\Gamma\setminus \left \lbrace v_\star \right \rbrace   } \d x_v.  
\end{align}
Before we analyze this integral in detail, the structure of the integrand \cref{Pgamma_factors4} immediately implies two useful lemmas.
\begin{lemma}\label{lem:odd_vanishes}
	If $\Gamma$ has an odd loop number $L$ (\cref{loop_number}), then $\alpha_\Gamma=0$.
\end{lemma}
\begin{proof}
	 The polynomial $W_\Gamma$ (\cref{Pgamma_factors4}) is homogeneous of degree $L$ in the variables $\vec x$. It might be that this polynomial contains even powers of individual variables $x_v$, but   the overall degree is odd. Therefore, every summand in $W_\Gamma$ contains at least one variable with an odd power, which means it will vanish when \emph{all} the $x_v$ are being integrated in a symmetric domain.
\end{proof}

\begin{lemma}\label{lem:factorization}
	Let $\Gamma_1$ and $\Gamma_2$ be distinct graphs.
	\begin{enumerate}
		\item If $\Gamma=\Gamma_1 \cup \Gamma_2$ is a disjoint union (i.e. the two graphs are not connected to each other), then $\alpha_\Gamma=0$.
		\item If $\Gamma= \Gamma_1 \circ_v \Gamma_2$ are connected at exactly one vertex $v$, then  $\alpha_\Gamma= \pm \alpha_{\Gamma_1}\wedge \alpha_{\Gamma_2}$.
		\item If $\Gamma= \Gamma_1 \cup e \cup \Gamma_2$ are connected  through exactly one edge $e$, then \\$\alpha_\Gamma= \pm 2 a_e~ \alpha_{\Gamma_1}\wedge \alpha_{\Gamma_2}$.
	\end{enumerate}
\end{lemma}
\begin{proof}
	1. By \cref{Pgamma_factors4}, the integrand $W_\Gamma$ is a sum over spanning trees. A disconnected graph has no spanning trees. 
	
	\medskip 
	\noindent
	2. Assume that $v_\star \in \Gamma_1$ for the combined graph $\Gamma$. If $\Gamma_1$ and $\Gamma_2$ are connected through only one vertex $v$, then every spanning tree emenating from $v_\star$ is  a spanning tree in $\Gamma_1$, and  it enters $\Gamma_2$ through $v$. Hence, the vertex $v\in \Gamma_2$ takes the role of $v_\star$ for $\Gamma_2$ in the sense that every spanning treee is directed towards $v$. Which spanning tree is chosen in $\Gamma_2$ is independent of the tree chosen in $\Gamma_1$, all combinations are possible. The sum over all spanning trees factorizes into a sum over all spanning trees in $\Gamma_1$ times a sum over all spanning trees in $\Gamma_2$. The integrand for $\alpha_\Gamma$ first contains all factors corresponding to $\Gamma_1$, and then all factors corresponding to $\Gamma_2$. To form $\alpha_\Gamma$ according to \cref{Pgamma_factors4}, we need to put all $\d x$ to the right and all $\d a$ to the left. That is, the factors $\d x$ of $\Gamma_1$ have to cross the factors $\d a$ of $\Gamma_2$. The number of these $\d a$ is even because otherwise, by \cref{lem:odd_vanishes}, $\alpha_\Gamma=0$. Crossing an even number of terms introduces no extra sign and we indeed obtain $\alpha_\Gamma$. However, sorting the $\d x$ might introduce a sign (joining two sequences which had already been sorted does not imply that the concatenation of the two sequences is sorted). 
	
	\medskip
	\noindent
	3. Analogous to point 2, but every spanning tree $T$ contains the edge $e$. Hence, there is an overall factor $2 a_e$ in the sum \cref{Pgamma_factors4}. The edge $e$ can never give rise to a $\d a_e$, and therefore the combinatorics of the wedge product is not altered by the presence of $e$. 
\end{proof}

By \cref{lem:odd_vanishes,lem:factorization}, in all what follows we will assume that $\Gamma$ is a 1PI, or 2-connected, graph of even loop order.

\subsection{Solving the Gaussian integral}

The well-known general formula for the solution of a $n$-fold Gaussian integral is 
\begin{align}\label{omega_integrated}
	  \idotsint  e^{-  \vec x^T \laplacian \vec x}\;    W_\Gamma(\vec x) \; \d^n \vec x  
	&=  \frac{ \pi ^{\frac n2}}{  \sqrt{\det \laplacian}}  \exp \left( \frac 14\left( \laplacian^{-1} \right) _{jk}  \partial_{x_j} \partial_{x_k} \right)    W_\Gamma(\vec x) \Big|_{\vec x = \vec 0}.
\end{align}
Recall that by \cref{def:Symanzik_polynomial}, the determinant of the Laplacian $\prod_e a_e \cdot \det \laplacian=\psi$ is proportional the Symanzik polynomial of $\Gamma$.

\begin{example}{Dunce's cap}\label{ex:dunces_cap_omega}
	For the dunce's cap, the Laplacian was given in \cref{ex:dunces_cap_exponential_factor}, this leads to the Symanzik polynomial
	\begin{align*}
	\psi &= \prod_e a_e \cdot \det \laplacian = a_1 a_3 + a_2 a_3 + a_1 a_4 + a_2 a_4 + a_3 a_4.
	\end{align*}
	We have computed $W_\Gamma$  in \cref{ex:W_computed} or in \cref{ex:dunces_cap_signs}. Using this in \cref{omega_integrated}, leaving out the factor $\pi^{\frac 2 2}$ but including the prefactors from \cref{omega_integral}, one finds
	\begin{align*}
		\alpha_\Gamma
		&=   \frac{a_4 \left( \d a_1 \wedge \d a_3 + \d a_2 \wedge \d a_3 \right) - a_3 \left( \d a_1 \wedge \d a_4 + \d a_2 \wedge \d a_4 \right)  + \left( a_1 + a_2  \right) \d a_3 \wedge \d a_4 }{ 8\left(   a_1 a_3 + a_2 a_3 + a_1 a_4 + a_2 a_4 + a_3 a_4  \right) ^{\frac 32}  } .
	\end{align*}
\end{example}

We write the polynomial that corresponds to a single spanning tree in \cref{Pgamma_factors4} as a sum over monomials,
\begin{align}\label{def:XTs}
	X_T:=\prod_{e \notin T} \left( \incidencematrix \vec x \right)_e  =\sum_{s=1}^{\leq 2^{L }} X_{T,s}.
\end{align}
The number of monomials can be smaller than $2^{L }$   because not necessarily all $x$'s are distinct, and one of them can be the special vertex whose coordinate is set to zero.
Note that $X_T$, and the individual $X_{T,s}$, can contain higher than first powers of some $x_v$.

\begin{example}{Dunce's cap}\label{ex:dunces_cap_XTj}
	The dunce's cap has five different polynomials $X_T$, corresponding to the five spanning trees we discussed already in \cref{ex:dunces_cap_trees}, namely
		\begin{align*}
		X_{\left \lbrace 1,3 \right \rbrace   }=\left( I \vec x \right) _2\left( I \vec x \right) _4 &= \left( x_1 -0 \right) \left( x_2-0 \right) = x_1 x_2    \\
		X_{\left \lbrace 1,4 \right \rbrace   }=\left( I \vec x \right) _2\left( I \vec x \right) _3&=   \left( x_1-0 \right) \left( x_2-0 \right)  = x_1 x_2  \\
		X_{\left \lbrace 1,2 \right \rbrace   }=\left( I \vec x \right) _3\left( I \vec x \right) _4 &= \left( x_2-0 \right) \left( x_2-0 \right) = x_2 x_2\\
		 X_{\left \lbrace 2,3 \right \rbrace   }=\left( I \vec x \right) _1\left( I \vec x \right) _4&= \left( x_1-x_2 \right) \left( x_2-0 \right) = x_1 x_2 - x_2 x_2\\
		 X_{\left \lbrace 2,4 \right \rbrace   }=\left( I \vec x \right) _1 \left( I \vec x \right) _3&= \left( x_1-x_2 \right) \left( x_2-0 \right)= x_1 x_2 - x_2 x_2  .
	\end{align*}
	In this case, only the last two $X_T$ have more than one summand $X_{T,s}$ (\cref{def:XTs}). This is coincidence because the graph is so small and one in the edges \emph{not} in the spanning tree always touches the zero vertex $x_\star=x_3=0$.  
\end{example}

The exponential function in \cref{omega_integrated},
\begin{align}\label{exp_L}
	\exp \left( \frac 14 \left( \laplacian^{-1} \right) _{jk}  \partial_{x_j} \partial_{x_k} \right)    W_\Gamma(\vec x) \Big|_{\vec x = \vec 0},
\end{align}
is  a priori a series of infinitely many terms, given by its series expansion.  The derivatives $\partial_{x_j}$ imply that only those summands are non-zero where a factor $x_j$ appears in $ W_\Gamma(\vec x)$. 
Every monomial $X_{T,j}$ is of degree $L $. Consequently, the only non-zero terms in the exponential function \cref{exp_L} are those with exactly $L $ derivatives. By \cref{lem:odd_vanishes}, $L $ is even. Only one summand of the exponential contributes to \cref{exp_L}, namely
\begin{align*}
	\frac{1}{\left( \frac L 2 \right) !} \left( \frac 14  \left( \laplacian^{-1} \right) _{jk}  \partial_{x_j} \partial_{x_k}  \right) ^{\frac {L } 2}=	\frac{1}{2^L \left( \frac L 2 \right) !} \left(   \left( \laplacian^{-1} \right) _{jk}  \partial_{x_j} \partial_{x_k}  \right) ^{\frac {L } 2}.
\end{align*}
By \cref{lem:inverse_Laplacian_Dodgson}, the $(j,k)$-coefficient of the inverse Laplacian is given by the Dodgson polynomial $\psi^{j,k}$ (\cref{def:Dodgson_polynomial}): 
\begin{align}\label{def:Aij}
	\left( \laplacian^{-1} \right) _{jk} = \frac{(-1)^{j+k} \psi^{j,k }}{\psi }.
\end{align}
 The  Laplacian $\laplacian$ is symmetric and so is $ \psi^{j+k}$  under exchange $j \leftrightarrow k$. The Symanzik polynomial $\psi $ is homogeneous of degree $L $, the $\psi^{j,k}$ are homogeneous of degree $L +1$.
 
\begin{example}{Dunce's cap, Dodgson polynomials}\label{ex:dunces_cap_Aij}
	For the Dunces's cap as introduced in \cref{ex:dunces_cap_svector}, vertex $v_3=v_\star$ is special, so there are only two non-special vertices which can appear as indices (equivalently, $W_\Gamma$ is a polynomial in the two variables $x_1$ and $x_2$). One has the Dodgson polynomials
	\begin{align*}
	\psi^{1,1  } &=  (a_1 a_3 + a_1 a_4 + a_3 a_4) a_2 = \psi_{\Gamma|v_1=v_\star=v_3},  \\
\psi^{1,2}=\psi^{2,1  } &=- a_2 a_3 a_4,   \\
\psi^{2,2 } &=   (a_1 + a_2) a_3 a_4 =  \psi_{\Gamma|v_2=v_\star=v_3}.
	\end{align*} 
	We see that when two indices coincide, the Dodgson polynomial reduces to Symanzik polynomial of a  graph where the corresponding vertex is identified with $v_\star$. This property is well known, but not relevant for our current work.
\end{example}

If we include the prefactor from \cref{omega_integral}, introduce $\psi=\det \laplacian \cdot \prod_e a_e$, and specify $n=\abs{V_\Gamma}-1$, the formula of \cref{omega_integrated} for a graph with even loop order $L$ becomes 
\begin{align}\label{omega_integrated2}
	\alpha_\Gamma &= \frac{\pi^{\frac {\abs{V_\Gamma}-1} 2}}{ 2^{L+\abs{E_\Gamma}} \left( \frac L 2 \right) ! \cdot  \prod_e a_e    \cdot   \psi ^{\frac{L+1}{2}}} 	  \left( \sum_{i,j\in V_\Gamma} (-1)^{i+j}\psi^{j,k}  \partial_{x_j} \partial_{x_k}  \right) ^{\frac {L } 2}  W_\Gamma(\vec x).
\end{align}
It is no longer necessary to specify $\vec x=\vec 0$ because after carrying out $L$ derivatives, the result is independent of $\vec x$. 

\subsection{Derivatives as permutations of factors}
The sum of derivatives appearing in   \cref{omega_integrated2} is very similar to terms that usually appear in Feynman integrals. It is well known that such sums of derivatives evaluate to a sum over all matchings of indices, as we will discuss now.

\begin{example}{Action of derivatives on a monomial}\label{ex:action_on_monomial}
	Consider a  monomial $X=x_1 x_2 x_3 x_4$, assuming all $x$'s are different. Such a monomial arises in  4-loop graphs. Evaluating the derivative once, we get
	\begin{align*}
	\left( A_{ij}\partial_i \partial_j \right) X &=   \left( A_{12} + A_{21} \right)  x_3 x_4 + \left( A_{13} + A_{31} \right) x_2 x_4 + \ldots, 
	\end{align*}
	with a total of 12 summands. Now, evaluate the remaining derivatives to obtain
	\begin{align*}
		\left( A_{ij}\partial_i \partial_j \right)^2 X &=   \left( A_{12} + A_{21} \right)  \left( A_{34} + A_{43} \right)  + \left( A_{13} + A_{31} \right) \left( A_{24} + A_{42} \right)  + \ldots.
	\end{align*}
	If we expand all parentheses, we obtain 24=4! summands. These summands can be identified with the 4! permutations of the factors in $X$. A partition of each permutation into sets of size 2 will determine the indices of $A$:
	\begin{align*}
	\left \lbrace x_1, x_2, x_3, x_4 \right \rbrace &\mapsto \left \lbrace \left \lbrace x_1,x_2 \right \rbrace ,\left \lbrace x_3,x_4 \right \rbrace     \right \rbrace \mapsto A_{12} A_{34}\\
	  \left \lbrace x_1, x_2, x_4, x_3 \right \rbrace &\mapsto \left \lbrace \left \lbrace x_1,x_2 \right \rbrace ,\left \lbrace x_4,x_3 \right \rbrace     \right \rbrace \mapsto A_{12} A_{43}\\
	  \left \lbrace x_1, x_3, x_2, x_4 \right \rbrace &\mapsto \left \lbrace \left \lbrace x_1,x_3 \right \rbrace ,\left \lbrace x_2,x_4 \right \rbrace     \right \rbrace \mapsto A_{13} A_{24}.
	\end{align*}
	If $A_{ij}= A_{ji}$ and $A_{ij}A_{kl}= A_{kl}A_{ij}$, only three are truly distinct:
		\begin{align*}
		\left( A_{ij}\partial_i \partial_j \right)^2 X &=   8 A_{12} A_{34} + 8 A_{13}A_{24} + 8 A_{14}A_{23}.
	\end{align*} 
	These are the familiar three possibilities for a perfect matching of four elements. One can easily verify that one would still obtain these 24 summands if not all four factors in $X$ were distinct.
\end{example}

\begin{example}{Dunce's cap}\label{ex:dunces_cap_Xj}
	Following \cref{ex:dunces_cap_XTj}, the dunce's cap   has only two distinct monomials $X_{T,j}$, namely $x_1x_2$ and $x_2^2$. Each of them gives rise to $2!=2$ terms, 
	\begin{align*}
	x_1 x_2 &\rightarrow \psi^{1,2}  + \psi^{2,1} =    2\psi^{1,2}\\
	x_2 x_2 &\rightarrow \psi^{22} + \psi^{2,2} = 2 \psi^{2,2}.
	\end{align*}
	In this simple example, every sum collapses to just a single term, in general, there will be a non-trivial sum remaining.
\end{example}

Going back to \cref{Pgamma_factors4}, a tree $T$ contributes to $\alpha_\Gamma$ with a term proportional to the $\d a_e$ for $e \notin T$. To extract this contribution, we define the coefficient extraction operator
\begin{align}\label{def:dbarT}
\left[ \d \bar T \right] \alpha_\Gamma &:= \left[ \bigwedge_{e \notin T} \d a_e \right] \alpha_\Gamma= \alpha_\Gamma \Big|_{\text{Terms proportional to }\bigwedge_{e \notin T} \d a_e}.
\end{align}
Combining this with \cref{omega_integrated2}, we find  
\begin{align}\label{dbarT}
\left[ \d \bar T \right] \alpha_\Gamma  
&=  \frac{\pi^{\frac {\abs{V_\Gamma}-1}2} \sgn(T) (-1)^{ \abs{V_\Gamma}-1}   \left(  \prod_{e \in T} 2 a_e  \right)}{ 2^{L+\abs{E_\Gamma}} \left( \frac L 2 \right) ! \cdot  \prod_e a_e  \cdot   \psi ^{\frac{L+1}{2}}}    \sum_s \left(  (-1)^{j+k} \psi^{j,k}  \partial_{x_j} \partial_{x_k}  \right) ^{\frac L 2}  X_{T,s} .
\end{align}
Here, $\sgn(T)$ is given by \cref{lem:sign_T}, and we have written the sign from the product of $(-2a_e)$ explicitly.
The individual summands are polynomials
\begin{align}\label{def:Yts}
Y_{T,s}&:= \left(  (-1)^{j+k}\psi^{j,k}  \partial_{x_j} \partial_{x_k}  \right) ^{\frac L 2}  X_{T,s}.
\end{align}
Each $Y_{T,s}$ is a polynomial in Schwinger parameters $a_e$, homogeneous of degree $(L+1)^{\frac L2}$. It inherits the sign of $X_{T,s}$, which, in turn, is determined from the corresponding term  in \cref{def:XTs}.

\begin{lemma}[Wick/Isserlis theorem]\label{lem:permutations} Fix $T$, then the contribution  	 $P_{T,s}$ (\cref{def:Yts}), contributing to the differential \cref{dbarT} of all edges $e\notin T$, is a sum of all permutations $\sigma \in S_L$ of its $L$ factors $x_i$,
	\begin{align*}
	Y_{T,s}  
	&= \textnormal{sgn}(X_{T,s}) (-1)^{\sum_{x_k\in X_{T,s} } k} \sum_{\sigma \in S_L} \psi^{  \sigma(x_{i_1})  , \sigma(x_{i_2}) } \cdots \psi^{   \sigma(x_{i_{L-1}}) , \sigma(x_{i_L})     }.  
	\end{align*}
\end{lemma}
\begin{proof}
	This statement is well-known, we have illustrated it in \cref{ex:action_on_monomial}. It is a general property of higher moments of Gaussian distributions, known as Isserlis' theorem \cite{isserlis_formula_1918}. The equivalent statement for higher moments of free quantum field theories is called Wick's theorem \cite{wick_evaluation_1950}. 
	
	Note that all terms in the sum on the right hand side in \cref{lem:permutations} appear with the same sign because they all involve the same set of indices, the factor $(-1)^{i+j}$ from \cref{def:Aij} is the same for all of them, and appears in front of the sum. 
\end{proof}

In \cref{lem:permutations}, due to the symmetry $\psi^{i,j}=\psi^{j,i}$, many terms are actually equal. Concretely, there are $\frac L 2$ factors $\psi^{i,j}$, and hence, summing over all permutations, at least $2^{\frac L2}$ of the summands are equal upon using the symmetry. Even more might be equal if some of the indices coincide. With \cref{lem:permutations}, \cref{def:dbarT} becomes
\begin{align}\label{omega_dE}
	\left[ \d \bar T \right] \alpha_\Gamma &=  \frac{\pi^{\frac {\abs{V_\Gamma}-1}2} \sgn(T) (-1)^{\abs{V_\Gamma}-1}  }{ 4^{L } \left( \frac L 2 \right) ! \cdot  \prod_{e \notin T} a_e  \cdot   \psi ^{\frac{L+1}{2}}}   	\\
	&\qquad \cdot  \sum_{s} \sgn \left( X_{T,s} \right)  (-1)^{\sum_{x_k\in X_{T,s} } k}  \sum_{\sigma \in S_L} \psi^{\sigma(x_{i_1}),\sigma(x_{i_2})} \cdots \psi^{\sigma(x_{i_{L-1}}),\sigma(x_{i_L})}. \nonumber 
\end{align}
Here, we have canceled the $\abs T = \abs{V_\Gamma}-1$ factors of 2 from the numerator against $L+\abs{E_\Gamma}= 2 \abs{E_\Gamma}-(\abs{V_\Gamma}-1)$ factors of 2 in the denominator. Moreover,   the monomial $X_{T,s}= \left \lbrace x_{i_1}, \ldots, x_{i_L} \right \rbrace   $ (\cref{def:XTs}) is  interpreted as a set of its $L$ factors, not all of which are necessarily distinct.

\begin{example}{Dunce's cap}
Consider the first spanning tree, $T= \left \lbrace 1,3 \right \rbrace $, of the dunce's cap in \cref{ex:dunces_cap_XTj}. There is only a single monomial $X_{T,1}=X_T= x_1 x_2$ with sign $\sgn(X_{T,1})=1$.  By \cref{ex:dunces_cap_Aij,ex:dunces_cap_Xj}, this term gives rise to 
\begin{align*}
\sum_{\sigma \in S_L} \psi^{\sigma(x_{i_1}),\sigma(x_{i_2})} \cdots \psi^{\sigma(x_{i_{L-1}}),\sigma(x_{i_L})}=2 \psi^{1,2} =  -2a_2 a_3 a_4.
\end{align*}
 The graph has $L=2$ loops, the sign  $\sgn(T)=-1$ has been determined in \cref{ex:dunces_cap_signs}, and $ (-1)^{\sum_{x_k\in X_{T,s} } k}= (-1)^{1+2}=-1$.    \Cref{omega_dE}, leaving out $\pi$, yields
\begin{align*}
\left[ \d \bar T \right] \alpha_\Gamma= \left[ \d a_2 \wedge \d a_4 \right] \alpha_\Gamma &=  \frac{  (-1)  }{4^2 1! a_2 a_4    \cdot   \psi^{\frac{2+1}{2}}} 1\cdot (-1) (-2a_2 a_3 a_4	) =  \frac{   - a_3    }{8    \psi^{\frac 32} }     	 .
\end{align*}
This coincides with the  coefficient from the manual computation in  \cref{ex:dunces_cap_omega}.
\end{example}

\subsection{Combinations of vertex Dodgson polynomials}
\Cref{omega_dE} holds for an individual tree $T$, and within this tree it is a sum over the individual summands $X_{T,s}$ of the polynomial $X_T$ (\cref{def:XTs}). These monomials, however, have additional structure, namely, we know that the vertices can only be one of the two end vertices of an edge. Phrased differently: All the possible vertices that make $X_{T,s}$ for a given $T$ actually correspond to \emph{one} set of edges. It turns out that this structure guarantees that the sum over $s$ in $X_{T,s}$ can be computed exactly in terms of Dodgson polynomials.
Finally, we arrive at the following expression for the parametric integrand $\alpha_\Gamma$:

\begin{theorem}\label{lem:dbarT}
	For a fixed tree $T$, the summand of $\d \bar T$ in $\alpha_\Gamma$ is
	\begin{align*} 
		\left[ \d \bar T \right] \alpha_\Gamma &=  \frac{    \pi^{\frac {\abs{V_\Gamma}-1}2}  (-1)^{\abs{V_\Gamma}-1}  }{ 4^{  L} \left( \frac L 2 \right) ! \cdot   \psi ^{\frac{L+1}{2}}}   
		\det \big( \incidencematrix[T] \big)  \sum_{\sigma \in S_L(\bar T)} \psi^{ \sigma(e_1) ,\sigma(e_2)    }\cdots \psi^{ \sigma(e_{L-1}), \sigma(e_L)   } .
	\end{align*}
	Here, $\psi^{i,j}$ are edge-indexed Dodgson polynomials (\cref{def:Dodgson_polynomial}) and the sum goes over all permutations $\sigma$ of the $L$ edges in $\bar T$, and these edges are assumed to be in increasing order as always.
\end{theorem}
\begin{proof}
	 The sum in  \cref{omega_dE} is a sum of all combinations of vertex-indexed Dodgson polynomials.  The number of edges not in a spanning tree equals the loop number $L$, which is even. Pick any two edges $e_1,e_2\notin T$ and select all those terms from the sum which are vertex Dodgson polynomials where one index is from $e_1$ and one is from $e_2$. Apply \cref{lem:dodgson_edge_edge} to these terms, they will combine to form an edge-Dodgson polynomial. As we are summing over \emph{all} combinations, the resulting expression gets multiplied by a sum over all vertex Dodgson polynomials that do not arise from $e_1,e_2$.  Keep selecting pairs of edges until all vertex Dodgson polynomials have been combined into edge Dodgson polynomials.  
	 
	 From \cref{lem:dodgson_edge_edge}, we obtain a factor $\prod_{e \notin T} a_e$ which cancels the corresponding factor from \cref{omega_dE}. 
	 
	 Furthermore, note that the sign $(-1)^{\sum_{x_k\in X_{T,s}}k}$ means that every $\psi^{x_i,x_j}$ is multiplied by $(-1)^{i+j}$  as required for \cref{lem:dodgson_edge_edge}, and the relative signs in the lemma are exactly the overall signs $\sgn(X_{T,s})$ that arise from the individual monomials according to \cref{def:XTs}. By \cref{lem:dodgson_edge_edge}, these signs combine to form
	 \begin{align*}
	 (-1)^{\sum_{e_k \notin T} k + \frac L 2}.
	 \end{align*}
	 But that is just the sign from \cref{sign_factor1}, hence $\sgn(T)$ from   \cref{lem:sign_T} reduces to $\det \left( \incidencematrix[T] \right) $ as claimed.
	 
\end{proof}

\begin{example}{Long 2-loop graph, disjoint vertices}\label{ex:dodgson_2loop}
	Consider the following large 2-loop graph: 
	
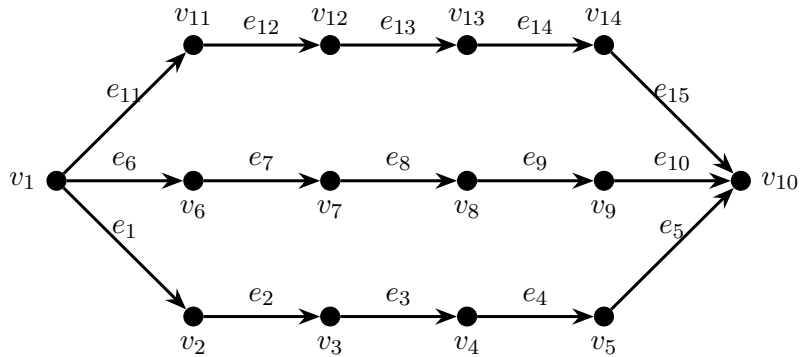
\begin{center}
		\begin{tikzpicture}[scale=1.8, >=Stealth]
			\node[vertex,label=left:$v_1$] (v1) at (0,0){};
			\node[vertex,label=below:$v_2$] (v2) at (1,-1){};
			\node[vertex,label=below:$v_3$] (v3) at (2,-1){};
			\node[vertex,label=below:$v_4$] (v4) at (3,-1){};
			\node[vertex,label=below:$v_5$] (v5) at (4,-1){};
			\node[vertex,label=below:$v_6$] (v6) at (1,0){};
			\node[vertex,label=below:$v_7$] (v7) at (2,0){};
			\node[vertex,label=below:$v_8$] (v8) at (3,0){};
			\node[vertex,label=below:$v_9$] (v9) at (4,0){};
			\node[vertex,label=right:$v_{10}$] (v10) at (5,0){};
			\node[vertex,label=above:$v_{11}$] (v11) at (1,1){};
			\node[vertex,label=above:$v_{12}$] (v12) at (2,1){};
			\node[vertex,label=above:$v_{13}$] (v13) at (3,1){};
			\node[vertex,label=above:$v_{14}$] (v14) at (4,1){};
			
			\draw[edge, ->](v1) to node[pos=.5, above]{$e_1$} (v2);
			\draw[edge, ->](v2) to node[pos=.5, above]{$e_2$} (v3);
			\draw[edge, ->](v3) to node[pos=.5, above]{$e_3$} (v4);
			\draw[edge, ->](v4) to node[pos=.5, above]{$e_4$} (v5);
			\draw[edge, ->](v5) to node[pos=.5, above]{$e_5$} (v10);
			\draw[edge, ->](v1) to node[pos=.5, above]{$e_6$} (v6);
			\draw[edge, ->](v6) to node[pos=.5, above]{$e_7$} (v7);
			\draw[edge, ->](v7) to node[pos=.5, above]{$e_8$} (v8);
			\draw[edge, ->](v8) to node[pos=.5, above]{$e_9$} (v9);
			\draw[edge, ->](v9) to node[pos=.5, above]{$e_{10}$} (v10);
			\draw[edge, ->](v1) to node[pos=.5, above]{$e_{11}$} (v11);
			\draw[edge, ->](v11) to node[pos=.5, above]{$e_{12}$} (v12);
			\draw[edge, ->](v12) to node[pos=.5, above]{$e_{13}$} (v13);
			\draw[edge, ->](v13) to node[pos=.5, above]{$e_{14}$} (v14);
			\draw[edge, ->](v14) to node[pos=.5, above]{$e_{15}$} (v10);
			
		\end{tikzpicture}
	\end{center}
	Let the  special vertex be $v_\star=v_{14}$. Choosing a spanning tree means to choose exactly two edges not in the tree. If we pick $\bar T := \left \lbrace e_2, e_8 \right \rbrace $,	then the polynomial 
	\begin{align*}
		X_T &= \left( x_3-x_2 \right) \left( x_8-x_7 \right)  = x_3 x_8 - x_2 x_8 - x_3 x_7 + x_2 x_7
	\end{align*}
	has four monomials, each of them gives two (identical) vertex Dodgson polynomials due to the exchange of the two factors, and the sum of the terms in \cref{lem:permutations} is
	\begin{align*}
		\sum_{s=1}^4 Y_{T,s} &=    (-1)^{3+8} 2\psi^{3,8  } -  (-1)^{2+8} 2\psi^{2,8 }      - (-1)^{3+7}2 \psi^{3,7}  +  (-1)^{2+7}2 \psi^{2,7 }    \\
		&=   -2 \psi^{e_3,e_8} -2 \psi^{e_2,e_8}      - 2 \psi^{e_3,e_7} -2 \psi^{e_2,e_7}.
	\end{align*}
 	 By \cref{lem:dodgson_edge_edge}, this sum simplifies considerably:
	\begin{align*}
		\sum_{s=1}^4 Y_{T,s} &=  2a_2 a_8 \left( a_{11}+a_{12} + a_{13} + a_{14} + a_{15} \right)    = -2 a_2 a_8 \psi^{e_2,e_8  }. 
	\end{align*}
	This is the result expected from \cref{lem:dbarT}, where the factor of 2 again arises from the sum over permutations $\psi^{e_2,e_8 }+\psi^{e_8,e_2}$, and the scalar prefactors of \cref{lem:dbarT} are not present since we only considered $\sum_s Y_{T,s}$ instead of $[\d \bar T]\alpha_\Gamma$. 
	
	The choice $\bar T = \left \lbrace e_2, e_8 \right \rbrace $ was special in that it neither involved the special vertex $v_\star=v_{14}$ nor any vertex twice. This restriction is not necessary. For example for $\bar T  := \left \lbrace e_{10},e_{15} \right \rbrace $, where both edges are adjacent to $v_{10}$, one obtains
	\begin{align*}
		\sum_{s=1}^2 Y_{T,j}		&=   2\left(  \psi^{v_{10},v_{10} } + \psi^{  v_9,v_{10} }      \right)   \\
		&= 2\left(  a_{10} a_{15} \left( a_{1} + a_{2} + a_{3} + a_{4} + a_{5} \right)   \right) 	= 2 a_{10} a_{15} \psi^{e_{10},e_{15}}. 
	\end{align*}
\end{example}

The edge-type Dodgson polynomials in \cref{lem:dbarT} are symmetric upon exchanging their arguments, this means that every summand appears $2^{\frac L 2}$ times. Moreover, the product is symmetric upon exchanging the factors, this gives another factor $\left( \frac L 2 \right) !$. Hence, 
\begin{align}\label{dT_distinct_summands}
\text{\# distinct summands } &= \frac{L!}{\left( \frac L 2 \right) ! 2^{\frac L 2}} = \frac{L!}{L!!}=(L-1)!!\\
&= \left \lbrace 1,~3,~15,~105,~945,~10395,~ \ldots \right \rbrace \quad \text{for}\quad L=\left \lbrace 2,4,6,8,\ldots  \right \rbrace   .  \nonumber
\end{align}
This makes it possible to compute the sum for typical 6-loop graphs relatively easily with a computer.

\section{Wedge product}\label{sec:wedge_product}

We will establish that
\begin{align*}
	\alpha_\Gamma \wedge \alpha_\Gamma=0
\end{align*}
for all 1PI graphs. This implies that $\alpha_\Gamma \wedge \alpha_\Gamma=0$ for all non-tree graphs because by \cref{lem:factorization}, if $\Gamma$ is not 1PI then $\alpha_\Gamma$ factorizes, and   if $\alpha_\Gamma= \alpha_1 \wedge \alpha_2$ then $\alpha_\Gamma \wedge \alpha_\Gamma = \pm \left( \alpha_1 \wedge \alpha_1 \right) \wedge \left( \alpha_2 \wedge \alpha_2 \right) $, which is zero if either of the 1PI factors is zero.

\begin{example}{Non-vanishing wedge product}\label{ex:wedge_non_vanishing}
	Notice, first of all, that the claim $\alpha_\Gamma \wedge \alpha_\Gamma=0$ is not trivial. 
	The wedge product is symmetric for even products of differentials, for example
	\begin{align*}
		&\Big( \d a_1 \wedge \d a_2 + \d a_3 \wedge \d a_4 \Big) \wedge \Big( \d a_1 \wedge \d a_2 + \d a_3 \wedge \d a_4 \Big) \\
		&= 2 \d a_1 \wedge \d a_2 \wedge \d a_3 \wedge \d a_4 \neq 0.
	\end{align*}
	Conversely, for an odd number of differentials, there is an overall minus sign when sorting the mixed terms and they cancel out each other:
	\begin{align*}
		\left( \d a_1 \wedge \d a_2 \wedge \d a_3 + \d a_4 \right) \wedge 	\left( \d a_1 \wedge \d a_2 \wedge \d a_3 + \d a_4   \right)&=0.
	\end{align*}
 The number of differentials $\d a$ in $\alpha_\Gamma$ equals the loop number (\cref{loop_number}), which is even by \cref{lem:odd_vanishes}, hence it would be conceivable that $\alpha_\Gamma \wedge \alpha_\Gamma \neq 0$.
	
\end{example}

\begin{example}{Dunce's cap}
	We know from \cref{ex:dunces_cap_omega} that $\alpha_\Gamma$ for the dunces cap equals
	\begin{align*}
	   \frac{a_4 \left( \d a_1 \wedge \d a_3 + \d a_2 \wedge \d a_3 \right) - a_3 \left( \d a_1 \wedge \d a_4 + \d a_2 \wedge \d a_4 \right)  + \left( a_1 + a_2  \right) \d a_3 \wedge \d a_4 }{8  \psi^{\frac 32} } .
	\end{align*}
	Indeed, explicit calculation shows that $\alpha_\Gamma \wedge \alpha_\Gamma=0$ in this case. Firstly, there are all those terms where any $\d a_e$ appears twice, they vanish automatically. What remains are exactly those combinations where the sets of edges are disjoint. These are
	\begin{align*}
		a_4 (-a_3) \left( \d a_1 \wedge \d a_3  \right)  \wedge \left( \d a_2 \wedge d a_4  \right)   + a_4 (- a_3) \left( \d a_2 \wedge \d a_3  \right) \wedge \left( \d a_1 \wedge d a_4 \right) =0.
	\end{align*}
\end{example}

Even if, by \cref{ex:wedge_non_vanishing}, the wedge product of an even number of differentials is not trivially zero, there are still some graphs with even loop order where the wedge product vanishes trivially, namely, if they do not permit sufficiently many distinct edge differentials.

\begin{lemma}\label{lem:edges_upper_bound}
	If $2L > \abs{E_\Gamma}$, or equivalently $\abs{E_\Gamma} > 2  \abs{V_\Gamma}-2$, then $\alpha_\Gamma \wedge \alpha_\Gamma=0$. In particular, if $\Gamma$ is a multiedge on 2 vertices, then $\alpha_\Gamma \wedge \alpha_\Gamma=0$.
\end{lemma}
\begin{proof}
	By \cref{dbarT}, each factor $\alpha_\Gamma$ contains as many $\d a_e$ as there are edges not in a spanning tree. These are $L $ (the loop number \cref{loop_number}) edges. 
	Consequently, a wedge product $\alpha_\Gamma \wedge \alpha_\Gamma$ contains $2 L  $ differentials $\d a_e$.
	By antisymmetry, the wedge product $\alpha_\Gamma \wedge \alpha_\Gamma$ vanishes unless all involved $\d a_e$ are distinct. Ignoring any subtleties, for this to be possible at all, the graph must have at least $2 L $ edges. However, at a fixed number of vertices, every additional edge increases the loop number, therefore, this can equivalently be viewed as an \emph{upper} bound on the number of edges.
\end{proof}
Note that the condition of \cref{lem:edges_upper_bound} does not restrict the remaining class of graphs too much. For example, for the special case of $r$-regular graphs (i.e. each vertex has valence $r$), the condition $\abs{E_\Gamma} \leq 2 \abs{V_\Gamma}-2$ is true exactly if $r \leq 4$, (these are the graphs that arise in the renormalizable scalar theories in 4-dimensions). As a second example, a $n$-Laman graph \cite[Sect.~3.3]{gaiotto_higher_2024} has $n \abs{V_\Gamma} = (n-1) \abs{E_\Gamma}+n+1$, and hence $\frac{n}{n-1}\abs{V_\Gamma}- \frac{n+1}{n-1} = \abs{E_\Gamma}$. These graphs satisfy $\abs{E_\Gamma}\leq  2 \abs{V_\Gamma}-2$ for all  $n \geq 2$. 
\bigskip

By \cref{lem:edges_upper_bound}, we concentrate on  $\abs{E_\Gamma}\leq 2 \left( \abs{V_\Gamma}-1 \right) $. 
For every spanning tree, there are $L$ edges not in the tree. If the bound \cref{lem:edges_upper_bound} is exhausted, then the two spanning trees are complements of each other. Otherwise, there is an overlap between the spanning trees, but still, the edge sets \emph{not} in the trees are disjoint.

Let now $E$ be the set of exactly $2L $ distinct  edges \emph{not} in the spanning trees, that is 
\begin{align}\label{def:E}
E = E_1 \cup E_2 = \bar T_1 \cup \bar T_2 = \left( E_\Gamma \setminus T_1 \right) \cup \left( E_\Gamma \setminus T_2 \right) = \left \lbrace e_{i_1} , \ldots, e_{i_{2L}} \right \rbrace   .
\end{align}
As in \cref{def:dbarT}, we let
\begin{align}\label{def:dE}
	\d E &:= \d a_{i_1} \wedge \ldots \wedge \d a_{i_{2L}}.
\end{align} 
The sought-after wedge product consists of multiple terms of this form, each multiplied by some polynomial $Q_E$ in the remaining edge variables $\left \lbrace a_e \right \rbrace $. 
\begin{align}\label{def:QE}
	\alpha_\Gamma \wedge \alpha_\Gamma &= \sum_{\substack{E \subseteq E_\Gamma \\ \abs{E}=2L}} \bar Q_E \; \d E.
\end{align}
Our task is to describe the $\bar Q_E$ and to establish that they vanish.

Every $\bar Q_E$ arises from a sum of \emph{all} ways that produce a differential $\d E$. For a fixed set $E$ of edges, this amounts to all ways of partitioning  $E$ into two disjoint parts according to \cref{def:E}.  There are exactly $\binom {2L}L$ ways to do that. Not for all of these partitions, the two sets describe spanning trees, but we ignore this restriction for now. We can still include those partition if we make sure that they give a zero contribution to the sum.

In the  set $\d E$ (\cref{def:dE}), all factors $\d a_e$ appeared in increasing order. For a given partition $E=E_1\cup E_2$, the factors of $E_1$, which belong to the first $\alpha_\Gamma$, come before the factors of $E_2$. Therefore, when computing the wedge product \cref{def:QE}, these factors will be shuffled and need to be brought into increasing order. On the other hand, $E_1$ and $E_2$ are individually sorted. The sign that arises from the shuffle is the sign of the \emph{partition}, $\sgn(E_1 \oplus E_2)$. Here we use the notation $E_1 \oplus E_2$ instead of $E_1 \cup E_2$ to indicate that this object is a concatenation of two disjoint sequences of edges (which keeps a notion of order), not just the union (which would be the same regardless of which partition had been used). Phrased differently, we are interpreting $E_1,E_2$ as \emph{words} (sequences with specified order) instead of sets (which have no notion of order). 

\begin{example}
{} Consider $E:= \left \lbrace e_1,e_2,e_3,e_4 \right \rbrace $, then $\d E = \d a_1 \wedge \d a_2 \wedge \d a_3 \wedge d a_4$, where the factors are sorted. In the product \cref{def:QE}, the differential $\d E$ can arise, for example, from 
	\begin{align*}
	\left( \d a_1 \wedge \d a_2 \right) \wedge \left( \d a_3 \wedge \d a_4 \right) = \d E.
	\end{align*}
	But it might as well arise in a product of the form
	\begin{align*}
	\left( \d a_1 \wedge \d a_3 \right) \wedge \left( \d a_2 \wedge \d a_4 \right) = (-1) \d E.
	\end{align*}
	Note that the factors within the parentheses, corresponding to $E_1$ and $E_2$, are sorted in all cases, but still in the second case, their order-preserving concatenation $E_1 \oplus  E_2=a_1 a_3 a_2 a_4$ is not   sorted. 
\end{example}

This sign, $\sgn(E_1 \oplus E_2)$, is $+1$ if nothing needs to be permuted, that is, if all letters of the word $E_1\oplus E_2$ come in increasing order. Otherwise, keep in mind that the words $E_1$ and $E_2$ individually are sorted in order to be used as $\bar T$ in \cref{lem:dbarT}. Therefore, the sign of the shuffle equals the sign of the permutation of the word $E_1\oplus E_2$,  
\begin{align*}
\sgn \left(E_1 \oplus E_2 \right) &= \sgn_\text{perm}\Big(  E_1   \oplus  E_2    \Big) \qquad \text{when $E_1$ and $E_2$ individually are sorted}.
\end{align*}
Conversely, if we allow $E_1$ and $E_2$ to be words with arbitrary order of their elements, we need a correction factor to put them into increasing order:
\begin{align}\label{partition_sign_2}
\sgn \left( E_1 \oplus E_2 \right) &= \sgn_\text{perm}\left( E_1 \oplus E_2 \right) \cdot \sgn_\text{perm}\left( E_1 \right) \cdot \sgn_\text{perm}\left( E_2 \right) .
\end{align}
Even if in \cref{lem:dbarT} the arguments $E_j$ must be sorted increasingly, we will see later that it is useful to allow non-sorted $E_j$ as well, which requires \cref{partition_sign_2}. 

Recall that, summing over all partitions, the contribution of a fixed set $E$ of edges to the wedge product \cref{def:QE} is
\begin{align*}
	\bar Q_E &= \sum_{E = E_1 \cup E_2} \sgn(E_1 \oplus E_2) \cdot \left[ \d E_1 \right] \alpha_\Gamma \cdot \left[ \d E_2 \right] \alpha_\Gamma . 
\end{align*}
In this formula, all edges in the subsets $E_1$ and $E_2$ are in increasing order such that we can apply \cref{lem:dbarT} without any further signs. This will give rise to an overall scalar factor consisting of powers of $\psi $ and other quantities. We leave  out this factor since it is irrelevant for the question whether $\bar Q_E$ vanishes. The non-trivial part of $\bar Q_E$ is
\begin{align}\label{QE_sum}
	\tilde Q_E &:= \sum_{E =E_1 \oplus E_2} \sgn(E_1 \oplus E_2)  \det \big( \incidencematrix(E_1,\emptyset) \big)    \det \big(\incidencematrix(E_2,\emptyset) \big)  \nonumber\\
	&\qquad \cdot  \sum_{\sigma \in S_L(E_1)} \psi^{ \sigma(e_1), \sigma(e_2)   }  \cdots \psi^{ \sigma(e_{L-1}), \sigma(e_L)   }  \nonumber\\
	&\qquad  \cdot  \sum_{\sigma \in S_L(E_2)} \psi^{ \sigma(e_1), \sigma(e_2)     }  \cdots \psi^{ \sigma(e_{L-1}), \sigma(e_L)  } . 
\end{align} 
The signs in this sum, arising from \cref{lem:sign_T}, are problematic because the determinants depend on the   choice of graph matrix for a given graph, that is, on the choice of labeling and direction of edges. In \cref{QE_sum}, we have introduced the notation $\incidencematrix(E_1,\emptyset)$ of \cref{def:minors} to make contact with \cref{lem:dodgson_polynomial_determinants}. Concretely, by \cref{Dodgson_product_trees}, this product of the determinants is a Dodgson polynomial,
\begin{align}\label{sign_product_dodgson}
\det \big(\incidencematrix(E_1,\emptyset) \big) \det \big(\incidencematrix(E_2,\emptyset ) \big) &=  \psi^{  E_1,   E_2}  \in \left \lbrace -1,0,+1 \right \rbrace . 
\end{align}
This identity holds even for the case where $E_1$ and $E_2$ are arbitrary sets of $L$ edges each, but the expression is non-zero only when both $E_1$ and $E_2$ are complements of spanning trees.  On the other hand, the Dodgson polynomial $ \psi^{E_1, E_2} $ is the determinant of the $(\abs{V_\Gamma}-1)\times (\abs{V_\Gamma}-1)$ minor $\mathbb M(E_1,E_2)\equiv \mathbb M[T_1,T_2]$. According to \cref{lem:dodgson_sign_determinant}, this Dodgson polynomial can be expanded as an alternating sum over Dodgson polynomials $\psi^{i,j} $, where $i$ and $j$ are individual edges:
	\begin{align}\label{dodgson_sign_sum}
	\psi^{E_1,E_2}  &= \frac{1}{\left( \psi  \right) ^{L-1}}\sum_{m\in M(E_1,E_2) } \sgn(m) \underbrace{\prod_{(p_1,p_2)\in m} \psi^{p_1,p_2} }_{L \textnormal{ factors}},
\end{align}
where $M(E_1,E_2)$ are the $L!$ matchings of the $L$-element sets $E_1$ and $E_2$. 
The Symanzik polynomial $(\psi)^{L-1}$ in \cref{dodgson_sign_sum} is yet another trivial overall factor, we factor it out of $\tilde Q_E$ (\Cref{QE_sum}) to form $Q_E$ according to   
\begin{align}\label{QE_sum2}
	Q_E &:=\sum_{E =E_1 \cup E_2} \sgn(E_1 \oplus E_2)  \sum_{m\in M(E_1,E_2)} \sgn(m) \prod_{(p_1,p_2)\in m}\psi^{p_1,p_2}   \nonumber\\
	&\qquad \cdot  \sum_{\sigma \in S_L(E_1)} \psi^{ \sigma(e_1) , \sigma(e_2)  } \cdots \psi^{  \sigma(e_{L-1})  ,   \sigma(e_L)  }  \nonumber\\
	&\qquad  \cdot  \sum_{\sigma \in S_L(E_2)} \psi^{\left \lbrace \sigma(e_1) \right \rbrace , \left \lbrace \sigma(e_2) \right \rbrace    } \cdots \psi^{ \sigma(e_{L-1})   ,   \sigma(e_L)   }  . 
\end{align}

\begin{example}{2-loop with general Dodgson polynomials}\label{ex:2loop_permutations}
	
	Consider an arbitrary 1PI 2-loop graph. Such a graph has the shape of the letter $\theta$, where the three branches might be arbitrarily long, compare the graph considered in \cref{ex:dodgson_2loop}.

	Without loss of generality, let $E= \left \lbrace 1,2,3,4 \right \rbrace $  be the $2L$ selected edges. Up to trivial exchange of the two components, there are three distinct partitions, they have the following signs:
	\begin{align*}
	\left \lbrace 1,2 \right \rbrace \oplus \left \lbrace 3,4 \right \rbrace &: \quad +1\\
	\left \lbrace 1,3 \right \rbrace \oplus \left \lbrace 2,4 \right \rbrace &:\quad -1\\
	\left \lbrace 1,4 \right \rbrace \oplus \left \lbrace 2,3 \right \rbrace &: \quad +1.      
	\end{align*}
	\Cref{QE_sum}, using \cref{sign_product_dodgson}, reads 
	\begin{align*}
		Q_E &= +1 \psi^{\left \lbrace 1,2 \right \rbrace ,\left \lbrace 3,4 \right \rbrace    } \psi^{1,2 } \psi^{3,4} - \psi^{\left \lbrace 1,3 \right \rbrace ,\left \lbrace 2,4 \right \rbrace    } \psi^{1,3} \psi^{2,4} + \psi^{\left \lbrace 1,4 \right \rbrace ,\left \lbrace 2,3 \right \rbrace    } \psi^{1,4}\psi^{2,3}.
	\end{align*}
	The Dodgson polynomial with two-component sets as indices are zero or $\pm 1$. In particular, they are zero unless both $E_1$ and $E_2$ are complements of spanning trees. Instead of examining the combinatorial possibilities of choosing spanning trees, we now use  \cref{dodgson_sign_sum} to expand these Dodgson polynomials, and we obtain \cref{QE_sum2}:
	\begin{align*}
		\psi\cdot Q_E &= \left( \psi^{1,3}\psi^{2,4}-\psi^{1,4}\psi^{2,3} \right) \psi^{1,2 } \psi^{3,4}\\
		&\qquad - \left( \psi^{1,2}\psi^{3,4}-\psi^{1,4}\psi^{2,3} \right)   \psi^{1,3} \psi^{2,4} \\
		&\qquad + \left( \psi^{1,2}\psi^{3,4}-\psi^{1,3}\psi^{2,4}\right)   \psi^{1,4}\psi^{2,3}\\
		&=  \psi^{1,2 } \psi^{1,3}\psi^{2,4}  \psi^{3,4}-  \psi^{1,2 } \psi^{1,4} \psi^{2,3}  \psi^{3,4}\\
		&\qquad -  \psi^{1,2} \psi^{1,3} \psi^{2,4}\psi^{3,4}  + \psi^{1,3}\psi^{1,4} \psi^{2,3}     \psi^{2,4} \\
		&\qquad + \psi^{1,2}\psi^{1,4}\psi^{2,3} \psi^{3,4}  -\psi^{1,3}\psi^{1,4}\psi^{2,3}\psi^{2,4}   \\
		&= 0.
	\end{align*}
	The sum vanishes regardless of the  values of the particular Dodgson polynomials.

\end{example}
Already at $L=4$, the sums of \cref{QE_sum2} are too complicated to do them by hand, but \cref{ex:2loop_permutations} is paradigmatic: The fact that the sum vanishes is due to combinatorial cancellations between the summands. We never need to evaluate any $\psi^{i,j}$ explicitly, it does not even matter if a given $\psi^{i,j}$ is zero or not. 

\begin{figure}[htbp]
	\centering
		\begin{tikzpicture}

			\coordinate(x0) at (0,0);
			\node at ($(x0)+(.5,1)$) {$+(1,2) \oplus (3,4)$};
			\node[vertex, fill=white, label=above:1](v1) at ($(x0)+(0,0)$) {};
			\node[vertex, fill=white, label=above:2](v2) at ($(x0)+(1,0)$) {};
			\node[vertex, label=below:3](v3) at ($(x0)+(0,-1)$) {};
			\node[vertex, label=below:4](v4) at ($(x0)+(1,-1)$) {};
			
			\draw[edge](v1)--(v2);
			\draw[edge](v3)--(v4);
			\draw[edge,dashed](v1)--(v3);
			\draw[edge,dashed](v2)--(v4);
			\node at ($(x0)+(.5,-2)$) {$+\psi^{1,3}\psi^{2,4}\psi^{1,2} \psi^{3,4}$};
			\node at ($(x0)+(.5,-3)$) {fix $\left \lbrace 1,4 \right \rbrace $ };
			\node at ($(x0)+(.5,-3.5)$) { exchange $2 \leftrightarrow 3$ };
			
			\coordinate(x0) at (5,0);
			\node at ($(x0)+(.5,1)$) {$-(1,2) \oplus (4,3)$};
			\node[vertex, fill=white, label=above:1](v1) at ($(x0)+(0,0)$) {};
			\node[vertex, fill=white, label=above:2](v2) at ($(x0)+(1,0)$) {};
			\node[vertex, label=below:3](v3) at ($(x0)+(0,-1)$) {};
			\node[vertex, label=below:4](v4) at ($(x0)+(1,-1)$) {};
			
			\draw[edge](v1)--(v2);
			\draw[edge](v3)--(v4);
			\draw[edge,dashed](v2)--(v3);
			\draw[edge,dashed](v1)--(v4);
			\node at ($(x0)+(.5,-2)$) {$-\psi^{1,4}\psi^{2,3}\psi^{1,2} \psi^{3,4}$};
			\node at ($(x0)+(.5,-3)$) {fix $\left \lbrace 1,3 \right \rbrace $ };
			\node at ($(x0)+(.5,-3.5)$) { exchange $2 \leftrightarrow 4$ };
			
			\coordinate(x0) at (10,0);
			\node at ($(x0)+(.5,1)$) {$+(1,3) \oplus (4,2)$};
			\node[vertex, fill=white, label=above:1](v1) at ($(x0)+(0,0)$) {};
			\node[vertex, label=above:2](v2) at ($(x0)+(1,0)$) {};
			\node[vertex, fill=white, label=below:3](v3) at ($(x0)+(0,-1)$) {};
			\node[vertex, label=below:4](v4) at ($(x0)+(1,-1)$) {};
			
			\draw[edge](v1)--(v3);
			\draw[edge](v2)--(v4);
			\draw[edge,dashed](v3)--(v2);
			\draw[edge,dashed](v1)--(v4);
			\node at ($(x0)+(.5,-2)$) {$+\psi^{1,4}\psi^{2,3}\psi^{1,3} \psi^{2,4}$};
			\node at ($(x0)+(.5,-3)$) {fix $\left \lbrace 1,2 \right \rbrace $ };
			\node at ($(x0)+(.5,-3.5)$) { exchange $3\leftrightarrow 4$ };

			\coordinate(x0) at (0,-5.5);	
			\node at ($(x0)+(.5,1)$) {$-(1,3) \oplus (2,4)$};
			\node[vertex, fill=white, label=above:1](v1) at ($(x0)+(0,0)$) {};
			\node[vertex, label=above:2](v2) at ($(x0)+(1,0)$) {};
			\node[vertex, fill=white, label=below:3](v3) at ($(x0)+(0,-1)$) {};
			\node[vertex, label=below:4](v4) at ($(x0)+(1,-1)$) {};
			
			\draw[edge](v1)--(v3);
			\draw[edge](v2)--(v4);
			\draw[edge,dashed](v1)--(v2);
			\draw[edge,dashed](v3)--(v4);
			\node at ($(x0)+(.5,-2)$) {$-\psi^{1,2}\psi^{3,4}\psi^{1,3} \psi^{2,4}$};

			\coordinate(x0) at (5,-5.5);
			\node at ($(x0)+(.5,1)$) {$+(1,4) \oplus (2,3)$};
			\node[vertex, fill=white, label=above:1](v1) at ($(x0)+(0,0)$) {};
			\node[vertex, label=above:2](v2) at ($(x0)+(1,0)$) {};
			\node[vertex,label=below:3](v3) at ($(x0)+(0,-1)$) {};
			\node[vertex,  fill=white, label=below:4](v4) at ($(x0)+(1,-1)$) {};
			
			\draw[edge](v1)--(v4);
			\draw[edge](v2)--(v3);
			\draw[edge,dashed](v1)--(v2);
			\draw[edge,dashed](v3)--(v4);
			\node at ($(x0)+(.5,-2)$) {$+\psi^{1,2}\psi^{3,4}\psi^{1,4} \psi^{2,3}$};
			
			\coordinate(x0) at (10,-5.5);
			\node at ($(x0)+(.5,1)$) {$-(1,4) \oplus (3,2)$};
			\node[vertex, fill=white, label=above:1](v1) at ($(x0)+(0,0)$) {};
			\node[vertex, label=above:2](v2) at ($(x0)+(1,0)$) {};
			\node[vertex,label=below:3](v3) at ($(x0)+(0,-1)$) {};
			\node[vertex,  fill=white, label=below:4](v4) at ($(x0)+(1,-1)$) {};
			
			\draw[edge](v1)--(v4);
			\draw[edge](v2)--(v3);
			\draw[edge,dashed](v4)--(v2);
			\draw[edge,dashed](v1)--(v3);
			\node at ($(x0)+(.5,-2)$) {$-\psi^{1,3}\psi^{2,4}\psi^{1,4} \psi^{2,3}$};

		\end{tikzpicture}
	\caption{Auxiliary graphs $G$ of  six distinct permutations $S$ of $E=\left \lbrace 1,2,3,4 \right \rbrace   $ for a 2-loop graph used in the proof of \cref{thm:QE_0}. These graphs correspond to  the six terms comprising $Q_E$ in \cref{ex:2loop_permutations}. \\ The first line of each figure shows the partition $\sgn_\text{perm}(e_1,e_2,e_3,e_4)\cdot (e_1,e_2)\oplus(e_3,e_4)$ with its sign   (\cref{proof_overall_sign}). Vertices of $E_1$ are drawn white, vertices of $E_2$ are black. 
	The dashed lines correspond to Dodgson polynomials in the first block of factors, \cref{factor1_product}. They go from $E_1$ to $E_2$, and their position is implied by the chosen permutation $S$. The bold lines represent the Dodgson polynomials of the second block of factors, \cref{factor2_product}. These are a perfect matchin \emph{within} either $E_1$ or $E_2$. In general, the bold edges are not fixed by $S$ alone, but for a 2-loop graph as shown here, they are. \\
	The cancellation used in \cref{thm:QE_0} occurs within each column. We fix two vertices at opposing sides,  and exchange the remaining vertices. In the 2 loop case, this means to exchange one pair of vertices. This exchanges dashed and solid lines, but reproduces the same summand in terms of Dodgson polynomials, just with a flipped sign.}
	\label{fig:visualization_2loop}
\end{figure}

\begin{theorem}\label{thm:QE_0}
	Let $\Gamma$ be a 1PI graph with $L$ loops where $L$ is even. Let $E \subseteq E_\Gamma$ with $\abs{E}=2L$ and let $Q_E$ be as in \cref{QE_sum2}. Then 
	\begin{align*}
	Q_E=0.
	\end{align*}
	This implies $\bar Q_E=0$ and therefore $\alpha_\Gamma \wedge \alpha_\Gamma=0$ by \cref{def:QE}.
\end{theorem}
\begin{proof}
	We examine the sum \cref{QE_sum2}. It consists of products of $2L $ edge-indexed Dodgson polynomials, coming with alternating signs. We will establish that these terms cancel pairwise. In what follows, we  assume that $E,E_1,E_2$ are \emph{words}, that is, sequences with a  fixed order of elements, not unordered sets. 
	
	\medskip 
	\noindent
	First, consider the sign. By \cref{partition_sign_2},
	\begin{align*}
		\sgn \left( E_1 \oplus E_2 \right) &= \sgn_\text{perm}\left( E_1 \oplus E_2 \right) \cdot \sgn_\text{perm}\left( E_1 \right) \cdot \sgn_\text{perm}\left( E_2 \right) .
	\end{align*}
	The other sign, $\sgn(m)$ from \cref{lem:dodgson_sign_determinant}, is the sign of the matching of the words $E_1$ and $E_2$. Phrased differently, this is the relative sign of their permutations, 
	\begin{align*}
	\sgn(m) &= \sgn_\text{perm} \left( E_1 \right) \cdot \sgn_\text{perm}\left( E_2 \right) .
	\end{align*}
	Hence, the overall sign of a given summand in $Q_E$ is 
	\begin{align}\label{proof_overall_sign}
			\sgn \left( E_1 \oplus E_2 \right) \sgn(m) &=\sgn_\text{perm}\left( E_1 \oplus E_2 \right).
	\end{align}
	This is the ordinary permutation sign for the permutation of the overall word $E$, where $E_1$ is the first $L$ elements of $E$ and $E_2$ is the second $L$ elements. 
	
	\medskip 
	\noindent
	Now we examine \emph{which} summands are present. The key to this is the observation that the sum over all partitions $E_1 \oplus E_2$, and then over all matchings of the individual $E_j$, is the same as a sum over all permutations $S$ of $E$, with the convention that the first $L$ elements, $\left \lbrace s_1, \ldots, s_L \right \rbrace $, form $E_1$, and the last $L$ elements, $\left \lbrace s_{L+1} , \ldots, s_{2L} \right \rbrace $,  form $E_2$. That is, $\left \lbrace s_1, s_2, \ldots, s_{2L-1},s_{2L} \right \rbrace = S=E_1 \oplus E_2$, and 
	\begin{align}\label{QE_sum3}
		Q_E &=\sum_{S\in S_{2L}(E)}  \sgn_\text{perm}(S)  \cdot   \psi^{s_1, s_{L+1}} \psi^{s_2, s_{L+2}} \cdots \psi^{s_L, s_{2L}}   \nonumber\\
		&\qquad \cdot  \sum_{\sigma \in S_L(E_1)}\psi^{ \sigma(s_1) , \sigma(s_2) } \cdots \psi^{  \sigma(s_{L-1})  ,   \sigma(s_L)  }  \nonumber\\
		&\qquad  \cdot  \sum_{\sigma \in S_L(E_2)}  \psi^{ \sigma( s_{L+1}) ,   \sigma(s_{L+2})     } \cdots \psi^{\sigma(s_{2L-1})   ,  \sigma(s_{2L})   }  . 
	\end{align}

	There are $\abs{E}=2L$   edge indices in $S$ and they are all distinct. We let $G$ be an auxiliary graph on $2L$ vertices, where each vertex is labeled with one of the indices of $E$ (i.e. the vertices of $G$ correspond to the selected edges of $\Gamma$). The bipartition $S=E_1\oplus E_2$ implies a coloring of the vertices of $G$. We draw the vertices white when they lie in $E_1$, and black when they are in $E_2$. 
	
	Each summand in \cref{QE_sum3} has $2L$ factors $\psi^{i,j}$, each of which involves exactly two distinct indices. These summands will be the edges in our auxiliary graph $G$. The Dodgson polynomial $\psi^{i,j}=\psi^{j,i}$ is symmetric under exchange of indices, hence the edges in $G$ are undirected. The first $L$  factors of a summand in \cref{QE_sum3} are
	\begin{align}\label{factor1_product}
	\psi^{s_1, s_{L+1}} \psi^{s_2, s_{L+2}} \cdots \psi^{s_L, s_{2L}} .
	\end{align}
	Each factor joins one index of $E_1$ to one index of $E_2$. Hence, in the auxiliary graph $G$, the factors in \cref{factor1_product} correspond to edges between a black and a white vertex. In fact, \cref{factor1_product} involves every index exactly once, hence, these edges are a matching of all black vertices of $G$ with all white vertices of $G$. We draw these edges as dashed lines in $G$. \Cref{fig:visualization_2loop} shows auxiliary graphs for the 2-loop case.
	
 	For a given permutation of $E$, the resulting auxiliary graph is unique. A fixed auxiliary graph can correspond to more than one permutation, namely, it corresponds to exactly all those permutations where the \emph{sets} $E_1,E_2$ are unaltered, but both \emph{words} are permuted in the same way such that the resulting matching is the same. In \cref{factor1_product}, those permutations  correspond to a permutation of the individual Dodgson polyonmials $\psi^{p_1,p_2}$, without changing the matching of indices. There are $L!$ such permutations, each of them amounts to permuting indices in $E_1$ and simultaneously in $E_2$, hence all these terms contribute with the same sign. Such permutations have no influence on the second block of factors in \cref{QE_sum3} because they leave the partition $S=E_1\oplus E_2$ fixed. Hence, we could multiply by $L!$ and ignore these permutations and demand without a loss of generality that the set $E_1$ should be sorted increasingly. The remaining partitions would then be in bijection with auxiliary graphs. However, we do not do this because the next step of the proof gets slightly easier when there is not an additional step of sorting indices involved. 
	
	The other block of $L$ factors in \cref{QE_sum3}, 
	\begin{align}\label{factor2_product}
\sum_{\sigma \in S_L(E_1)}\psi^{ \sigma(s_1) , \sigma(s_2) } \cdots \psi^{  \sigma(s_{L-1})  ,   \sigma(s_L)  }    \cdot   \sum_{\sigma \in S_L(E_2)} \psi^{ \sigma( s_{L+1}) ,   \sigma(s_{L+2})     } \cdots \psi^{\sigma(s_{2L-1})   ,  \sigma(s_{2L})   } 
	\end{align}
	represents the sum of all possibilities to choose Dodgson polynomials \emph{within} $E_1$ and $E_2$, respectively. In the auxiliary graph, these are   perfect matchings of the vertices $E_1$, and of the vertices $E_2$, respectively. We draw these factors as solid edges in $G$. Note that \cref{factor2_product} involves an additional summation over permutations, the matching is not implied by the summation over permutations $S=E_1 \oplus E_2$. Phrased differently: The permutation $S$ fixes the matching   \cref{factor1_product} \emph{between} $E_1$ and $E_2$, and for each fixed $S$, we need to sum over all perfect matchings \emph{within} $E_1$ and $E_2$. An example for this is shown in \cref{fig:visualization_4loop}.
	
	\medskip
	\noindent
	By construction, the auxiliary graph has $2L$ vertices, $2L$ edges, and every vertex is incident to exactly one solid edge and one dashed edge. Hence, the auxiliary graph is either  a cycle, or a union of disjoint cycles. No two edges of the first type or of the second type are adjacent to each other, therefore they must alternate, therefore each of these cycles contains an even number of edges. Moreover, in order to match vertices of the same color with the black edges, the number of vertices of a given color must be even within each cycle (this is equivalent to saying that the loop number is even, \cref{lem:odd_vanishes}). Hence, each cycle has an even number of black vertices and the same even number of white vertices, and the total number of vertices is a multiple of four. For 2 loops, the situation is shown in \cref{fig:visualization_2loop}. In that case the auxiliary graph is just \emph{one} cylce because an even cycle can not have less than 4 edges. For 4 loops, it is possible to have either one or two cycles, see \cref{fig:visualization_4loop}.
	
	\medskip 
	\noindent
	We will now show that the terms in \cref{QE_sum3} cancel pairwise.
	Consider a cycle of   length $4k$, where $k$ is integer. Choose arbitrary one of the vertices to be at position 0 in the cycle, this fixes a vertex at position $2k$. Now exchange the labels of all other vertices pairwise \enquote{across the connection line between the two fixed vertices}, that is, exchange number 1 with number $4k-1$, number 2 with $4k-2$ and so on until number $2k-1$ with $2k+1$.  After this operation, the cycle has the same structure, the same vertex labels are present, the two edge types are still alternating,  but the labeling of vertices has changed. Equivalently, this operation exchanges the two  types of edge within the selected cycle, that is, it moves factors of Dodgson polynomials from the first block (\cref{factor1_product}) to the second (\cref{factor2_product}) and vice versa. This reordering of factors does not change the value of a summand because each Dodgson polynomial with the same indices is equivalent, no matter if it came from the first or the second block of factors. In doing this transformation, a total of $4k-2$ vertices are exchanged, these are $2k-1$ swaps of indices. This number is always odd. Consequently, the permutation sign \cref{proof_overall_sign} flips and we have obtained a summand in $Q_E$ which has all the same factors as the original summand, but a flipped overall sign. Hence, the two summands cancel, the same transformation can be done with \emph{every} summand, and $Q_E=0$.
\end{proof}

\begin{figure}[htb]
	\centering
	\begin{tikzpicture}

		\coordinate(x0) at (0,0);
		\node at ($(x0)+(0,2)$) {$+(1,5,2,7) \oplus (3,4,6,8)$};
		\node[vertex, fill=white, label=right:1](v1) at ($(x0)+(0:1)$) {};
		\node[vertex, fill=white, label=right:2](v2) at ($(x0)+(45:1)$) {};
		\node[vertex, label=above:3](v3) at ($(x0)+(90:1)$) {};
		\node[vertex, label=left:4](v4) at ($(x0)+(135:1)$) {};
		\node[vertex, fill=white, label=left:5](v5) at ($(x0)+(180:1)$) {};
		\node[vertex,  label=left:6](v6) at ($(x0)+(225:1)$) {};
		\node[vertex,fill=white, label=below:7](v7) at ($(x0)+(270:1)$) {};
		\node[vertex, label=right:8](v8) at ($(x0)+(315:1)$) {};
		
		\draw[edge,dashed](v1)--(v3);
		\draw[edge,dashed](v5)--(v4);
		\draw[edge,dashed](v2)--(v6);
		\draw[edge,dashed](v7)--(v8);
		\draw[edge](v1)--(v7);
		\draw[edge](v2)--(v5);
		\draw[edge](v3)--(v4);
		\draw[edge](v6)--(v8);
		\node at ($(x0)+(0,-2)$) {$+\psi^{1,3}\psi^{4,5}\psi^{2,6} \psi^{7,8} \psi^{1,7}\psi^{2,5}\psi^{3,4} \psi^{6,8}$};
		\node at ($(x0)+(0,-3)$) {fix $\left \lbrace 1,2 \right \rbrace $ };
		\node at ($(x0)+(0,-3.5)$) { exchange $3 \leftrightarrow 7$ \quad $4 \leftrightarrow 8$,\quad $5 \leftrightarrow 6$  };
		
		\coordinate(x0) at (8,0);
		\node at ($(x0)+(0,2)$) {$+(1,5,2,7) \oplus (3,4,6,8)$};
		\node[vertex, fill=white, label=right:1](v1) at ($(x0)+(0:1)$) {};
		\node[vertex, fill=white, label=right:2](v2) at ($(x0)+(45:1)$) {};
		\node[vertex, label=above:3](v3) at ($(x0)+(90:1)$) {};
		\node[vertex, label=left:4](v4) at ($(x0)+(135:1)$) {};
		\node[vertex, fill=white, label=left:5](v5) at ($(x0)+(180:1)$) {};
		\node[vertex,  label=left:6](v6) at ($(x0)+(225:1)$) {};
		\node[vertex,fill=white, label=below:7](v7) at ($(x0)+(270:1)$) {};
		\node[vertex, label=right:8](v8) at ($(x0)+(315:1)$) {};
		
		\draw[edge,dashed](v1)--(v3);
		\draw[edge,dashed](v5)--(v4);
		\draw[edge,dashed](v2)--(v6);
		\draw[edge,dashed](v7)--(v8);
		\draw[edge](v1)--(v5);
		\draw[edge](v2)--(v7);
		\draw[edge](v3)--(v4);
		\draw[edge](v6)--(v8);
		\node at ($(x0)+(0,-2)$) {$+\psi^{1,3}\psi^{4,5}\psi^{2,6} \psi^{7,8} \psi^{1,5}\psi^{2,7}\psi^{3,4} \psi^{6,8}$};
		\node at ($(x0)+(0,-3)$) {fix $\left \lbrace 1,4 \right \rbrace $ };
		\node at ($(x0)+(0,-3.5)$) { exchange $3 \leftrightarrow 5$   };

		\coordinate(x0) at (0,-6.5);	
		\node at ($(x0)+(0,2)$) {$-(1,6,2,3) \oplus (7,8,5,4)$};
		\node[vertex, fill=white, label=right:1](v1) at ($(x0)+(0:1)$) {};
		\node[vertex, fill=white, label=right:2](v2) at ($(x0)+(45:1)$) {};
		\node[vertex,fill=white, label=above:3](v3) at ($(x0)+(90:1)$) {};
		\node[vertex, label=left:4](v4) at ($(x0)+(135:1)$) {};
		\node[vertex,  label=left:5](v5) at ($(x0)+(180:1)$) {};
		\node[vertex, fill=white, label=left:6](v6) at ($(x0)+(225:1)$) {};
		\node[vertex,  label=below:7](v7) at ($(x0)+(270:1)$) {};
		\node[vertex, label=right:8](v8) at ($(x0)+(315:1)$) {};
		
		\draw[edge,dashed](v1)--(v7);
		\draw[edge,dashed](v6)--(v8);
		\draw[edge,dashed](v2)--(v5);
		\draw[edge,dashed](v3)--(v4);
		\draw[edge](v1)--(v3);
		\draw[edge](v2)--(v6);
		\draw[edge](v7)--(v8);
		\draw[edge](v4)--(v5);
		\node at ($(x0)+(0,-2)$) {$-\psi^{1,7}\psi^{6,8}\psi^{2,5} \psi^{3,4} \psi^{1,3}\psi^{2,6}\psi^{7,8} \psi^{4,5}$};

		\coordinate(x0) at (8,-6.5);
		\node at ($(x0)+(0,2)$) {$-(1,3,2,7) \oplus (5,4,6,8)$};
		\node[vertex, fill=white, label=right:1](v1) at ($(x0)+(0:1)$) {};
		\node[vertex, fill=white, label=right:2](v2) at ($(x0)+(45:1)$) {};
		\node[vertex,  fill=white, label=above:3](v3) at ($(x0)+(90:1)$) {};
		\node[vertex, label=left:4](v4) at ($(x0)+(135:1)$) {};
		\node[vertex, label=left:5](v5) at ($(x0)+(180:1)$) {};
		\node[vertex,  label=left:6](v6) at ($(x0)+(225:1)$) {};
		\node[vertex,fill=white, label=below:7](v7) at ($(x0)+(270:1)$) {};
		\node[vertex, label=right:8](v8) at ($(x0)+(315:1)$) {};
		
		\draw[edge,dashed](v1)--(v5);
		\draw[edge,dashed](v3)--(v4);
		\draw[edge,dashed](v2)--(v6);
		\draw[edge,dashed](v7)--(v8);
		\draw[edge](v1)--(v3);
		\draw[edge](v2)--(v7);
		\draw[edge](v5)--(v4);
		\draw[edge](v6)--(v8);
		\node at ($(x0)+(0,-2)$) {$-\psi^{1,5}\psi^{3,4}\psi^{2,6} \psi^{7,8} \psi^{1,3}\psi^{2,7}\psi^{4,5} \psi^{6,8}$};

	\end{tikzpicture}
	\caption{Auxiliary graphs for the cancellation of a summand in $Q_E$ for a 4-loop graph. Both columns show the same permutation $S$ of \cref{QE_sum3} (hence the coloring of vertices and the location of dashed lines in the upper row is the same), but they differ in the matching chosen in the second factor, \cref{factor2_product}. On the left side, the edges form one cycle of length 8. The vertex opposite to 1 is 2, and we exchange three pairs to obtain a canceling term.  On the right side, the auxiliary graph consists of two disjoint cycles, we only manipulate the cycle containing vertex index 1.}
	\label{fig:visualization_4loop}
\end{figure}
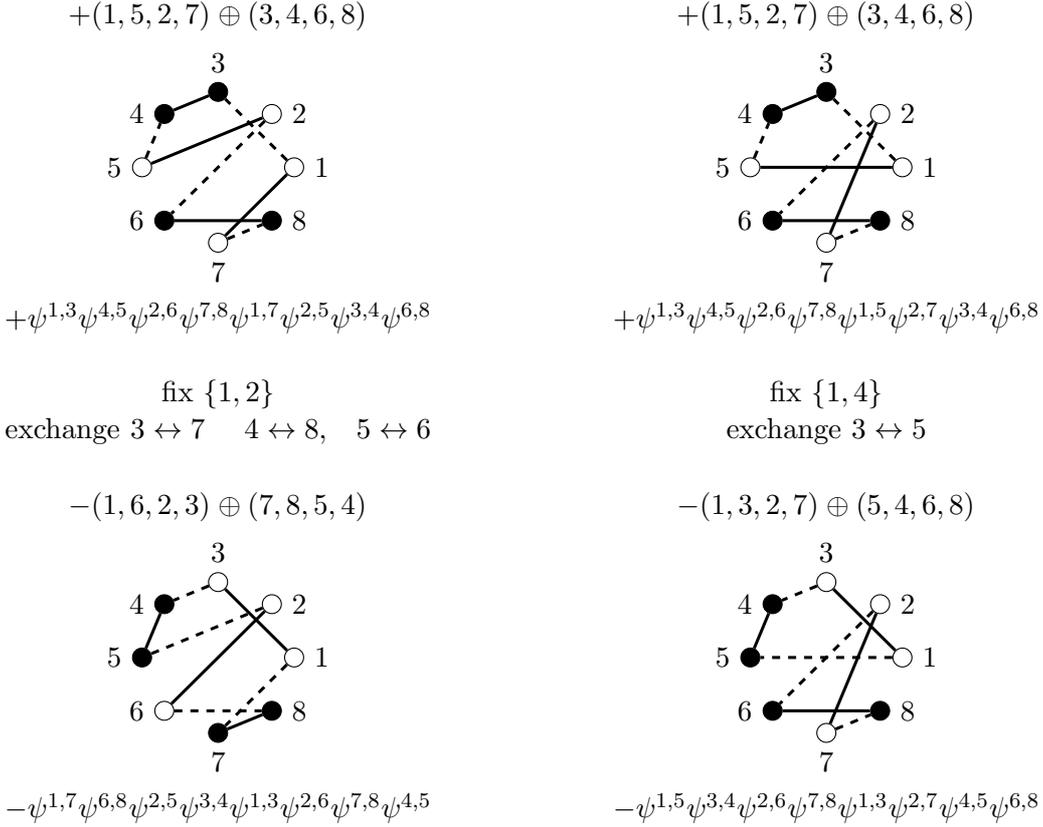

\newpage

\appendix

\section{Graph matrices}\label{sec:graph_matrices}

Let $\Gamma$ be a graph where edges $e\in E_\Gamma$ are directed, $e= \left( v^-(e) \rightarrow  v^+(e) \right) $ .

\begin{definition}\label{def:incidence_matrix}
	The \emph{incidence matrix} of $\Gamma$ is a matrix $\bar {\incidencematrix}$ with $\abs{E_\Gamma}$ rows and $\abs{V_\Gamma}$ columns. Its entry $\bar I_{e,v}$ has the value $-1$ if the edge $e$ starts at vertex $v$, or value $+1$ if $e$ ends at $v$. The entry is zero if $e$ is not incident to $v$.
\end{definition}

The matrix $\bar {\incidencematrix}^T \bar {\incidencematrix}$ is a $\abs{V_\Gamma}\times \abs{V_\Gamma}$ matrix where both rows and columns are indexed the vertices of the graph, this is the \emph{unlabeled graph Laplacian}. 
We want to be able to identify particular edges. To this end, we introduce one parameter $a_e$ for every edge $e$. In physics, these are the Schwinger parameters. From these parameters, we define the \enquote{edge variable matrix} 
\begin{align}\label{def:edge_variable_matrix}
	\edgematrix := \operatorname{diag}(\vec a)=\begin{pmatrix}
		a_1 & \\
		& \ddots & \\
		&& a_{\abs{E_\Gamma}}
	\end{pmatrix}.
\end{align}
Powers of the matrix $\edgematrix$ are understood in terms of the matrix exponential, but they effectively just amount to the corresponding power of the diagonal entries. 
\begin{align*}
	\edgematrix^{-1} = \operatorname{diag}\left(  \vec {(a^{-1})} \right).
\end{align*}
One now obtains the \emph{labeled graph Laplacian} $\bar{\laplacian}:=\bar{\incidencematrix}^T \edgematrix^{-1} \bar{\incidencematrix}$. However, the matrix $\bar{\laplacian}$ is not invertible. One therefore defines a variant of the above matrices where one arbitrary vertex $v_\star$ is excluded. To fix notation, we will always choose to remove the last vertex $v_\star=v_{\abs{V_\Gamma}}$, but all results are independent of this choice.
\begin{definition}\label{def:reduced_incidence_matrix}
	The \emph{reduced   incidence matrix} $\incidencematrix$ of $\Gamma$ amounts to the incidence matrix $\bar{\incidencematrix}$ (\cref{def:incidence_matrix}), but with the last column, corresponding to $v_\star$, left out. 
\end{definition}

\begin{definition}\label{def:Laplacian}
 	The  (reduced, labeled) \emph{Laplacian} of a graph $\Gamma$ is the matrix
 	\begin{align*}
 	 \laplacian &:=  {\incidencematrix}^T \edgematrix^{-1}  {\incidencematrix} =\mathbb  L^T.
 	\end{align*}
\end{definition}
\Cref{def:Laplacian} is sometimes called \emph{dual Laplacian} because the edge $e$ is labelled with the inverse parameter, $\frac{1}{a_e}$. This is merely convention. Conversions between the dual and non-dual version of the definitions is always possible by trivial operations.

\begin{definition}\label{def:expanded_Laplacian}
	The \emph{expanded Laplacian} is a matrix consisting of the reduced incidence matrix $\incidencematrix$ (\cref{def:reduced_incidence_matrix}) and the edge matrix (\cref{def:edge_variable_matrix}).
	\begin{align*}
		\mathbb M &:= \begin{pmatrix}
			\edgematrix & \incidencematrix \\
			-\incidencematrix^T &0
		\end{pmatrix}.
	\end{align*}
\end{definition}
\noindent
One finds that the determinants of $\laplacian$ and $\mathbb M$ agree up to an overall factor.
\begin{definition}\label{def:Symanzik_polynomial}
	The  \emph{Symanzik polynomial} of the graph $\Gamma$ is
	\begin{align*}
	\psi_\Gamma:=\det \mathbb M = \det \laplacian \cdot \prod_{e\in E_\Gamma} a_e.
	\end{align*}
\end{definition}
The Symanzik polynomial is alternatively given by the sum of all spanning trees of $\Gamma$, such that each monomial in $\psi_\Gamma$ contains the edges not in the tree.

\section{Dodgson polynomials}\label{sec:Dodgson_polynomials}

Consider sets $A $ and $B$ of integers between 1 and $(\abs{E_\Gamma}+\abs{V_\Gamma}-1)$, not necessarily disjoint. These sets correspond to sets of rows and columns of $\mathbb M$ (\cref{def:expanded_Laplacian}), we define two types of minors: 
\begin{align}\label{def:minors}
	\mathbb M  ( A,B  ) &:= \mathbb M \text{ where rows $A$ and colums $B$ have been removed}\\
	\mathbb M [A,B ] &:= \mathbb M \text{ where only rows $A$ and colums $B$ are present}. \nonumber 
\end{align}
If $\abs{A}=\abs{B}$, then these matrices are square and hence it makes sense to compute their determinant.

\begin{definition}\label{def:Dodgson_polynomial}
	Let $\abs{A}=\abs{B}$. The \emph{Dodgson polynomial} of a graph $\Gamma$ is the determinant of the minor of $\mathbb M$ (\cref{def:expanded_Laplacian}) where rows $A$ and columns $B$ have been removed,
	\begin{align*}
		\psi^{A,B}_\Gamma &= \det \mathbb M  (A,B).
	\end{align*}
	
\end{definition}
We will often leave out the index $_\Gamma$ from graph polynomials in order to not clutter notation.
Classically, the definition of Dodgson polynomials assumes that $A$ and $B$ are subsets of edge indices, but the definition and many properties are analogous for vertices \cite{golz_dodgson_2019}. For our applications, it is sufficient to restrict to the case where $A$ and $B$ are either both sets of edges, or both sets of vertices. The Dodgson polynomial is symmetric $\psi^{A,B}=\psi^{B,A}$ as long as $A$ and $B$ are either both sets of edges or both sets of vertices, but $\psi^{e,v}=-\psi^{v,e}$ for exchanging a single edge with a single vertex. 

Minors,   determinants and Dodgson polynomials have numerous interesting properties and relations, see for example \cite{vein_determinants_1999,brown_periods_2010,brown_spanning_2011,golz_dodgson_2019}. We will here review the ones that are relevant for our application, and prove new ones we need in the main text. 

In the special case where $A$ and $B$ consist of only one element, we will leave out braces, $\psi^{i,j} := \psi^{\left \lbrace i \right \rbrace ,\left \lbrace j \right \rbrace    }_\Gamma$. In that case, the Dodgson polynomials (\cref{def:Dodgson_polynomial}) coincide with the cofactors of $\mathbb M$. We recall that the \emph{adjungate} of the matrix $\mathbb M$ is the matrix consisting of all cofactors, 
\begin{align}\label{def:adjoint}
\operatorname{adj}(\mathbb M) &= \begin{pmatrix}
	\psi^{1,1} & -\psi^{1,2} & \ldots \\
	-\psi^{2,1} & \psi^{2,2} & \ldots\\
	\vdots & \vdots & \ddots 
\end{pmatrix}
\end{align}
The adjungate matrix encodes, up to an overall factor $\det(\mathbb M)$, the inverse matrix $\mathbb M^{-1}$. Due to the special structure of $\mathbb M$, the bottom right block in $\mathbb M^{-1}$ coincides with $\mathbb L^{-1}$.  
\begin{lemma}\label{lem:inverse_Laplacian_Dodgson}
	The $(i,j)$-entry of the inverse matrix of the graph Laplacian (\cref{def:Laplacian}), where $i$ and $j$ are indices of vertices, is given by Dodgson (\cref{def:Dodgson_polynomial}) and the Symanzik polynomial (\cref{def:Symanzik_polynomial}) according to 
	\begin{align*}
	\left( \laplacian^{-1} \right) _{ij} &= (-1)^{i+j} \frac{\psi^{i,j}}{\psi} .
	\end{align*}
\end{lemma}
Note that in \cref{lem:inverse_Laplacian_Dodgson} the indices $i,j$ of the Dodgson polynomial refer to vertices, not edges. To compute these polynomials from the matrix $\mathbb M$ in  \cref{def:expanded_Laplacian}, the indices need to be shifted by $\abs{E_\Gamma}$. 

\bigskip
Due to their definition as determinants of minors of $\mathbb M$ (\cref{def:Dodgson_polynomial}), the Dodgson polynomials satisfy various \emph{Dodgson identities} which typically arise from equating different ways to expand these determinants. In the main text, we need the following statement. 

\begin{lemma}\label{lem:dodgson_sign_determinant}
	Let $A,B\subset E_\Gamma$  with  $\abs{A}=\abs{B}=n$. Then
	\begin{align*}
		\psi^{A,B}  &= \frac{1}{\left( \psi_\Gamma  \right) ^{n-1}}\sum_{m } \sgn(m) \underbrace{\prod_{(p_1,p_2)\in m} \psi^{p_1,p_2}_\Gamma}_{n \textnormal{ factors}},
	\end{align*}
	where $m$ denotes one of the $n!$ possible matchings of the indices from $A$ with the indices of  $B$ in pairs $p$. The indices in $A,B$ are assumed in their unique increasing order. If the matching $m$ matches the $i$\textsuperscript{th} entry of $A$ to the $i$\textsuperscript{th} entry of $B$, then $\sgn(m)=1$, otherwise $\sgn(m)$ is the relative sign of the permutation of the matching.
\end{lemma}
\begin{proof}
	Use Jacobi's determinant formula \cite{jacobi_formatione_1841}. Let $(\operatorname{adj} \mathbb M)[A,B]$ be the adjungate (\cref{def:adjoint}) of the expanded graph Laplacian (\cref{def:expanded_Laplacian}) where only rows $A$ and columns $B$ are present. Let   $\mathbb M(A,B)$ be the expanded Laplacian where rows $A$ and columns $B$ have been deleted. Then 
	\begin{align}\label{proof_det}
		\det \Big( (\operatorname{adj}\mathbb M)[A,B] \Big) &= \Big( \det(\mathbb M) \Big) ^{\abs{A}-1} \det \big( \mathbb  M(A,B) \big) . 
	\end{align}
	By \cref{def:Symanzik_polynomial}, $\Big( \det(\mathbb M) \Big) ^{\abs{A}-1}=\psi^{n-1}$ is a power of the Symanzik polynomial. By \cref{def:Dodgson_polynomial}, $\det \big( \mathbb  M(A,B) \big)= \psi^{A,B}$ is the Dodgson polynomial. By \cref{def:adjoint}, the entries of the adjungate are the Dodgson polynomials $\psi^{i,j}$, where $i\in A$ and $j\in B$. The claimed formula follows if one applies Laplace expansion to the determinant on the left of \cref{proof_det}.
\end{proof}

Let $A\subset E_\Gamma$ be such that the minor  $\incidencematrix(A,\emptyset) $ (\cref{def:minors}) of the reduced incidence matrix (\cref{def:reduced_incidence_matrix}) is square, that is, $\abs{A}=\abs{E_\Gamma}-\abs{V_\Gamma}+1=L_\Gamma$. It is a standard fact of graph theory that the determinant of this minor is $\pm 1$ if the edges $E_\Gamma \setminus A=:T$ are  a spanning tree of $G$, and zero otherwise. Eqivalently, $\det \big( \incidencematrix[T]\big)$ is $\pm 1$ if and only if $T$ is a spanning tree, and zero else.
\begin{lemma}[\cite{brown_spanning_2011}]\label{lem:dodgson_polynomial_determinants}
	Let $A,B$ with $A \cap B=\emptyset$ be sets of edges. Let $\mathcal T$ be the set of subgraphs of $\Gamma \setminus(A \cup B)$ which have $\abs R=L_\Gamma-\abs{A}$ edges. Then the Dodgson polynomial (\cref{def:Dodgson_polynomial}) is an alternating sum over edges not in   spanning trees,
	\begin{align*}
		\psi^{A,B}_\Gamma &= \sum_{R \in \mathcal T}  \det \Big( \incidencematrix(R \cup A ,\emptyset)\Big) \cdot \det \Big( \incidencematrix(R \cup B ,\emptyset) \Big)\prod_{e \in R } a_e.
	\end{align*}
	Only those terms contribute where both  $R\cup A$ and $R\cup B$ are   spanning tree complements in $\Gamma$.
\end{lemma}
The signs in  the sum in \cref{lem:dodgson_polynomial_determinants} change when a different labeling or direction of the edges or vertices is chosen for the same graph. There have been efforts to introduce polynomials which do not have this ambiguity, such as   the \emph{spanning forest polynomials} \cite{brown_spanning_2011} or the \emph{Dodgson cycle polynomials}   \cite{golz_new_2017,golz_dodgson_2019}. For our application,   Dodgson polynomials are more suitable because of their well-defined relation to the inverse of the Laplacian in \cref{lem:inverse_Laplacian_Dodgson}.

The total number of edges minus the number of edges in a spanning tree is the loop number $L_\Gamma$ (=first Betti number, \cref{loop_number}) of a graph. 
In the special case where $\abs{A}=\abs{B}= L_\Gamma$, the set $R$ in \cref{lem:dodgson_polynomial_determinants} must have size zero. Hence, the sum has only a single  summand, $R=\emptyset$, and
\begin{align}\label{Dodgson_product_trees}
	\psi^{A,B}  &= \det \Big( \incidencematrix(A,\emptyset)\Big) \cdot \det \Big( \incidencematrix (B,\emptyset)\Big).
\end{align}
If both $A$ and $B$ are complements of spanning trees, this quantity is $+1$ or $-1$, otherwise it is zero.

\bigskip 
One crucial step for our analysis is relating the vertex-type Dodgson polynomials of \cref{lem:inverse_Laplacian_Dodgson} to edge-type Dodgson polynomials.

\begin{lemma}\label{lem:dodgson_vertex_edge}
	Let $e=v_1 \rightarrow v_2$ be an edge in $G$, and let $v\in V_\Gamma$ be not the removed vertex $v_\star$, but $v$ is allowed to coincide with either $v_1$ or $v_2$.
	Then
	\begin{align*}
		- \left( -1 \right) ^{v_1}\psi^{  v_1, v } + \left( -1 \right) ^{v_2}\psi^{ v_2 , v   }   &= \left( -1 \right) ^{e+\abs{E}} a_e \psi^{e,v  } .
	\end{align*}
	If $v_1$ or $v_2$ coincide with $v_\star$, the equation stays true if one leaves out the summand involving $v_\star$.
\end{lemma}
\begin{proof}
	Note that an edge $e=v_1 \rightarrow v_2$ implies that the incidence matrix (\cref{def:reduced_incidence_matrix}) has entries $\incidencematrix_{e, v_1}=-1$ and $\incidencematrix_{e,v_2}=+1$. We assume for now that $v_1<v_2$, hence  row $v_1$ comes above $v_2$. For our proof, it will be crucial to understand the column $e$ of the expanded   Laplacian $\mathbb M$ (\cref{def:expanded_Laplacian}), it has the form 
	\begin{align*}
\mathbb M = \begin{pmatrix}
	\edgematrix  & \incidencematrix\\
	-\incidencematrix^T &  0 
\end{pmatrix}=	\quad  \begin{matrix}
\vspace{.2cm} \\
~(\text{row }e )\\
\vspace{.1cm} \\
(\text{row }v_1) \\
\\
(\text{row } v_2)
\end{matrix}\begin{pmatrix}
		\ddots &   &   & & \quad  \incidencematrix\quad \\
		 & a_e &   &  & \vdots   \\
		&   &  \ddots & & \vdots \\
		\ldots & -(-1) & \ldots  &   &  \\
		\ldots &   & \ldots &  & \\
		\ldots & -(+1) & \ldots  & &
	\end{pmatrix}
	\end{align*}
	All empty entries in this matrix are zero. 
	
	Use \cref{def:Dodgson_polynomial} to compute the Dodgson polynomials. In   $\mathbb M(e,v)$,  row $e$ has been removed. Consequently, the diagonal in the upper left block is missing the entry $a_e$ and column $e$ has only two non-zero entries, namely $1$ at row $v_1$ and $-1$ at row $v_2$.  We compute $\psi^{e,v}$ by expanding the determinant with respect to column $e$:
	\begin{align*}
		\det \big( \mathbb M(e,v) \big)   &= 1 (-1)^{\abs{E_\Gamma} -1 + v_1 +e }\det \big(\mathbb M (\left \lbrace e,v_1 \right \rbrace  ,\left \lbrace e,v \right \rbrace    )\big)\\
		&\qquad 
		+(-1)(-1)^{\abs{E_\Gamma}-1  + v_2 +e }   \det \big(\mathbb M (\left \lbrace e,v_2 \right \rbrace  ,\left \lbrace e,v \right \rbrace    )\big)\\
		(-1)^{e+\abs {E_\Gamma} } a_e  \psi^{  e, v   }  &=  -(-1)^{  v_1 } a_e \psi^{\left \lbrace e,v_1 \right \rbrace  ,\left \lbrace e,v \right \rbrace    } 
		+(-1)^{ v_2 }  a_e \psi^{\left \lbrace e,v_2 \right \rbrace  ,\left \lbrace e,v \right \rbrace }. 
	\end{align*}
	
	Now consider the left hand side of the equation in \cref{lem:dodgson_vertex_edge}. We compute the determinant   $\det \left( \mathbb M(v_1,v) \right)  $ by expanding with respect to  column $e$. Having removed row  $v_1$,  column $e$ has only two non-vanishing entries left. We have $v_1<v_2$, so row $v_2$ lies below $v_1$ and the index of $v_2$ gets shifted by one if $v_1$ is removed. 
	\begin{align*}
		\det \mathbb M(v_1,v)  &= a_e (-1)^{e+e} \det\big( \mathbb M (\left \lbrace e,v_1 \right \rbrace , \left \lbrace e,v \right \rbrace  )\big)\\
		&\qquad 	+(-1)(-1)^{\abs{E}-1 +v_2+e}  \det \big(\mathbb M (\left \lbrace v_1,v_2 \right \rbrace  ,\left \lbrace e,v \right \rbrace )\big) \\
		(-1)^{v_1}\psi^{ v_1 , v  }   &= (-1)^{v_1} a_e \psi^{\left \lbrace e,v_1 \right \rbrace , \left \lbrace e,v \right \rbrace    } 
		+ (-1)^{\abs{E} +v_2+v_1+e}   \psi^{ \left \lbrace v_1,v_2 \right \rbrace  ,\left \lbrace e,v \right \rbrace } .
	\end{align*}
	The remaining term is expanded analogously, but because $v_1<v_2$, removing $v_2$ does not alter the row index of $v_1$. However, the entry in row $v_1$ is $1$ instead of $(-1)$. 
	\begin{align*}
		\det \big(\mathbb M(v_2,v)\big) &=  a_e (-1)^{e+e}\det \big(\mathbb M(\left \lbrace e,v_2 \right \rbrace , \left \lbrace e,v \right \rbrace    )\big) \\
		&\qquad 		+ 1(-1)^{\abs{E_\Gamma}  +v_1+e} \det \big(\mathbb M\left( \left \lbrace v_1,v_2 \right \rbrace  ,\left \lbrace e,v \right \rbrace      \right) \big) \\
		(-1)^{v_2}\psi^{\left \lbrace v_2 \right \rbrace ,\left \lbrace v \right \rbrace    }  &= (-1)^{v_2} a_e \psi^{\left \lbrace e,v_2 \right \rbrace , \left \lbrace e,v \right \rbrace   } 
		+ (-1)^{\abs{E}  +v_1+v_2+e} \psi^{ \left \lbrace v_1,v_2 \right \rbrace  ,\left \lbrace e,v \right \rbrace     } 
	\end{align*}
	Adding the two terms on the left hand side of \cref{lem:dodgson_vertex_edge}, we find
	\begin{align*}
		& - \left( -1 \right) ^{v_1}\psi^{ v_1 , v   } 
		+ \left( -1 \right) ^{v_2}\psi^{  v_2 , v   }   \\ 
		&= -(-1)^{v_1} a_e \psi^{\left \lbrace e,v_1 \right \rbrace , \left \lbrace e,v \right \rbrace    } 
		+ (-1)^{\abs{E_\Gamma} +v_2+v_1+e}   \psi^{ \left \lbrace v_1,v_2 \right \rbrace  ,\left \lbrace e,v \right \rbrace } \\
		&\qquad  +(-1)^{v_2} a_e \psi^{\left \lbrace e,v_2 \right \rbrace , \left \lbrace e,v \right \rbrace   } 
		- (-1)^{\abs{E_\Gamma}  +v_1+v_2+e} \psi^{ \left \lbrace v_1,v_2 \right \rbrace  ,\left \lbrace e,v \right \rbrace     }_\Gamma  \\
		&= -(-1)^{v_1} a_e \psi^{\left \lbrace e,v_1 \right \rbrace , \left \lbrace e,v \right \rbrace    }  +(-1)^{v_2} a_e \psi^{\left \lbrace e,v_2 \right \rbrace , \left \lbrace e,v \right \rbrace   } =	(-1)^{e+\abs {E_\Gamma} } a_e  \psi^{  e, v  } .
	\end{align*}
If we had assumed $v_2<v_1$, the signs of two terms would have changed, but the end result would have been the same. 
\end{proof}

\Cref{lem:dodgson_vertex_edge} involves an arbitrary vertex $v$, but  the same statement holds also for an arbitrary edge $e_2\neq e$ in place of $v$. 
We will need \cref{lem:dodgson_vertex_edge} with one particular combination of indices, namely for a product
\begin{align*}
 \left( v_2-v_1 \right) \left(  v_4-v_3 \right) = v_2 v_4 - v_1 v_4 + v_1 v_3 - v_2 v_3.
\end{align*}

\begin{lemma}\label{lem:dodgson_edge_edge}
	Let $e_1 = v_1 \rightarrow v_2$ and $e_2 = v_3 \rightarrow v_4$ be distinct edges. Then 
	\begin{align*}
		&(-1)^{v_2 + v_4} \psi^{  v_2, v_4     } - (-1)^{v_1+v_4} \psi^{ v_1 , v_4   } + (-1)^{v_1 + v_3} \psi^{ v_1 , v_3 } - (-1)^{v_2+v_3} \psi^{ v_2, v_3  }\\
		&=(-1)^{e_1+e_2+1 } a_{e_1} a_{e_2} \cdot  \psi^{  e_1, e_2   } .
	\end{align*}
\end{lemma}
\begin{proof}
Apply \cref{lem:dodgson_vertex_edge} :
	\begin{align*}
		&(-1)^{v_2 + v_4} \psi^{  v_2 , v_4  } - (-1)^{v_1+v_4} \psi^{  v_1 ,v_4  } + (-1)^{v_1 + v_3} \psi^{  v_1, v_3  } - (-1)^{v_2+v_3} \psi^{  v_2,v_3  }\\
		&= (-1)^{e_1+\abs{E_\Gamma}+v_4} a_{e_1} \psi^{  e_1, v_4    } - (-1)^{e_1 + \abs{E_\Gamma}+v_3} a_{e_1} \psi^{  e_1 , v_3   }\\
		&=-(-1)^{e_1+\abs{E_\Gamma} } a_{e_1} \big(- (-1)^{ v_3}   \psi^{  v_3 , e_1 }+ (-1)^{v_4}   \psi^{  v_4 , e_1  }  \big)  .
	\end{align*}
	We have used that the Dodgson polynomial with one edge index and one vertex index flips sign when the indices are exchanged. 
	The two vertices $v_3 \rightarrow v_4$ are the edge $e_2$, hence we use \cref{lem:dodgson_vertex_edge} again, but this time with the edge $e_1$ in place of the vertex $v$. 
	\begin{align*}
		&(-1)^{v_2 + v_4} \psi^{ v_2,v_4  } - (-1)^{v_1+v_4} \psi^{ v_1 , v_4   } + (-1)^{v_1 + v_3} \psi^{ v_1, v_3  } - (-1)^{v_2+v_3} \psi^{ v_2 , v_3  }\\
		&=-(-1)^{e_1+\abs{E_\Gamma} } a_{e_1} (-1)^{e_2+\abs{E_\Gamma}}  a_{e_2} \psi^{ e_1,e_2 }   =(-1)^{e_1+e_2+1 } a_{e_1} a_{e_2} \psi^{  e_1, e_2  } .
	\end{align*}
\end{proof}

\section{Scalar Feynman rules in parametric space} \label{sec:parametric_feynman_integral}

Many of the constructions in our proof are familiar from the derivation of parametric Feynman rules. This form of the Feynman rules have been known at least since the 1960s, more details and derivations can be found in  \cite{nakanishi_general_1957,bogoliubow_ueber_1957,panzer_feynman_2015,balduf_dyson_2024}. In the present appendix, we briefly review the steps to derive parametric Feynman rules for scalar fields.

We start from a momentum-space Feynman integral in $D$ spacetime dimensions, where we assume that the vertex Feynman rules are constant factors (such as $-i\lambda$), which we leave out.  
\begin{align}\label{feynman_integral_momentum}
	\feyk{\Gamma}&=  \left( \prod_{l\in L_\Gamma} \int \frac{\d^D \; k_l}{(2\pi)^D} \right) \left( \prod_{e \in E_\Gamma}  \frac{1}{(\vv k_e^2-m_e^2)^{\nu_e}} \right) . 
\end{align}
The parentheses in this expression are unnecessary, but kept for clarity. The loop momenta $  k_l$ do not necessarily coincide with edge momenta $k_e$, but they are linear functions of each other. Explicitly, they are related by momentum conservation at each vertex. Let $e\sim v$ denote the edges incident to a vertex $v$, and let $q_v$ be an external momentum at that vertex (which is zero when the vertex is internal), then  
\begin{align}\label{momenta_incidence_matrix}
0 &=\sum_{e \sim v} \vv k_e + q_v = \sum_{e \in E_\Gamma} \incidencematrix_{e,v}   k_e + \vv q_v.
\end{align}
Here, $\incidencematrix$ is the  incidence matrix (\cref{def:incidence_matrix}). We can write \cref{momenta_incidence_matrix} as a matrix-vector-product by introducint the $\abs{V_\Gamma}$-elements vector $\vec q$ of external momenta, and the $\abs{E_\Gamma}$-element vector $\vec k$ of edge momenta. The condition \cref{momenta_incidence_matrix} amounts to a product of $\abs{V_\Gamma}$ delta functions, which we write as Fourier integrals over auxiliary $D$-vectors $x_v$:
\begin{align*}
	\left( \prod_{v\in V_\Gamma} (2 \pi)^D \delta\left(  \incidencematrix^T \vec k + \vec q\right)_v  \right)  &= \left( \prod_{v\in V_\Gamma}\int \limits_{-\infty}^\infty \d^D x_v  \right) e^{i  \vec k^T \incidencematrix  \vec x+  i \vec q^T \vec x}.
\end{align*}
With this delta function, the momentum-integrals in the Feynman integral of \cref{feynman_integral_momentum} can be taken about all edge momenta, not just the independent loop momenta. We obtain a \enquote{hybrid} form between momentum-space and position-space (indeed, we could as well have started from the position-space formula, but starting from momentum space, it is more natural to consider external \emph{momenta} and not fixed external \emph{positions} of vertices).
\begin{align}\label{feynman_integral_hybrid}
	\feyk{\Gamma}&=  \left( \prod_{e\in E_\Gamma} \int \frac{\d^D \; k_e}{(2\pi)^D} \right) \left( \prod_{v\in V_\Gamma}\int \limits_{-\infty}^\infty \d^D \vv x_v  \right) \left( \prod_{e \in E_\Gamma}  \frac{1}{(\vv k_e^2-m_e^2)^{\nu_e}} \right) e^{i  \vec k^T \tilde \incidencematrix  \vec x+  i \vec q^T \vec x}  . 
\end{align}
One of the delta functions expresses overall momentum conservation, by convention, this delta function is not included in $\feyk{\Gamma}$. Therefore, we pick one vertex $v_\star$ and integrate only over the remaining vertices. From now on, we assume that the vectors $\vec x, \vec q$ only contain $(\abs{V_\Gamma}-1)$ remaining elements, and $\incidencematrix$ is the reduced incidence matrix (\cref{def:reduced_incidence_matrix}).

Now, we transform each of the propagators to Schwinger parametric space by using 
\begin{align*}
	\frac{1}{\left( \vv k_e^2 -m_e^2 \right) ^{\nu_e} } &=\int \limits_0^\infty  \frac{\d a_e \, a_e^{\nu_e-1}}{\Gamma(\nu_e)}  e^{-a_e (\vv k_e^2-m_e^2)}.
\end{align*}
Here, $\Gamma(\nu_e)$ is the Euler Gamma function. 
The product of these expressions give rise to an exponent which can be written as a matrix-vector product between the vector $\vec k$ of edge momenta, the vector $\vec m$ of edge masses, and the diagonal matrix $\edgematrix$ (\cref{def:edge_variable_matrix}) of edge variables $a_e$.
The Feynman rules are then 
\begin{align*}
	\feyk{\Gamma} 
	&=\left( \prod_{e \in E_\Gamma} \int \limits_0^\infty \frac{\d a_e \, a_e^{\nu_e-1}}{\Gamma(\nu_e)} \right)  \left( \prod_{e\in E_\Gamma} \int \frac{\d^D \vv k_e}{(2\pi)^D} \right)\left( \prod_{v\in (V_\Gamma\setminus v_0)} \int \d^D \vv x_v \right)  \\
	&\qquad \exp \left( 	-\vec k^T \edgematrix \vec k + \vec m^T \edgematrix \vec m + i \vec k^T   \incidencematrix \vec x + i \vec x^T \vec q.   \right).
\end{align*}
We want to integrate over both $\vec x$ and $\vec k$, we start with $\vec k$. This is a Gaussian integral where the variable $\vec k$ needs to be shifted by $-\frac i 2 \edgematrix^{-1} \incidencematrix \vec x$ in order to complete the square. Solving the integral, one obtains a power of $\det \edgematrix=\prod_{e\in E_\Gamma} a_e$ in the denominator. We ignore powers of $2\pi$, they can be absorbed by appropriate conventions. The result is 
\begin{align*}
	\feyk{\Gamma}
	&=  \left( \prod_{e \in E_\Gamma} \int \limits_0^\infty \frac{\d a_e \, a_e^{\nu_e-1}}{\Gamma(\nu_e)} \right) \left( \prod_{v\in (V_\Gamma\setminus v_0)} \int \d^D \vv x_v \right) \frac{1}{  \left( \det \edgematrix \right) ^{\frac D 2}} \\
	&\qquad e^{ 	-\frac 1 4  \vec x  ^T \incidencematrix^T \edgematrix^{-1}  \incidencematrix \vec x   + i \vec x^T \vec q + \vec m^T\edgematrix\vec m }   .
\end{align*}
In the exponent, we recognize the Laplacian $\laplacian = \incidencematrix^T\edgematrix^{-1}\incidencematrix$ (\cref{def:Laplacian}). The remaining integral over $\vec x$ is again Gaussian, it produces a power of $\det  \laplacian$ in the denominator.  
\begin{align}\label{feynman_rules_parametric}
	\feyk{\Gamma}
	&=  \left( \prod_{e \in E_\Gamma} \int \limits_0^\infty \frac{\d a_e \, a_e^{\nu_e-1}}{\Gamma(\nu_e)} \right)  \frac{1}{   \left( \det \edgematrix \right) ^{\frac D 2}  \left( \det \laplacian  \right) ^{\frac D 2}  } e^{ - \vec q ^T   \laplacian^{-1}   \vec q    + \vec m^T \edgematrix \vec m }   .
\end{align}
The   denominator is given by  the first Symanzik polynomial (\cref{def:Symanzik_polynomial}), $\psi = \det \edgematrix \det \laplacian$,
and the remaining exponent is defined as the second Symanzik polynomial, 
\begin{align*}
	\phi  &:= \psi  \cdot  \vec q ^T  \laplacian^{-1}  \vec q   - \psi  \cdot \sum_{e \in E_\Gamma} a_e m_e^2  . 
\end{align*}
The entries of the inverse Laplacian are vertex-indexed Dodgson polynomials according to \cref{lem:inverse_Laplacian_Dodgson}. These Dodgson polynomials, up to sign, are spanning 2-forest polynomials \cite{brown_spanning_2011}, the second Symanzik polynomial is a sum over spanning 2-forests of $\Gamma$, see e.g. \cite{bogner_feynman_2010} for details. 

The construction in the main text, in particular \cref{sec:integration}, is analogous to the derivation of \cref{feynman_rules_parametric}. We compute an integral over $\vec x$ and obtain $\alpha_\Gamma$, which is the integrand for parametric Feynman rules (i.e. it is a differential form in $\d a_e$). The distinctive difference from the present scalar case is that the integration in the main text involves not just the exponential $e^{-\vec x \laplacian \vec x}$, but   additionally the non-trivial polynomial $W_\Gamma(\vec x)$ defined in \cref{Pgamma_factors4}. Solving a Gaussian integral with polynomial integrand inevitably leads to factors $\laplacian^{-1}$ in the result, hence, the result $\alpha_\Gamma$ is a function of  Dodgson polynomials.

\medskip 

Finally, we remark that in the literature, the term \enquote{parametric Feynman integral} often means a slightly different expression than \cref{feynman_rules_parametric}. That version can be obtained by using that  both Symanzik polynomials are homogeneous, which allows to introduce an overall scaling of all Schwinger parameters $a_e \mapsto t \cdot a_e$. One can then solve the $t$-integral analytically, it produces an Euler gamma function, and the remaining integral is over the \emph{projective} space of Schwinger parameters, where one can e.g. fix the sum to be unity:
\begin{align}\label{feynman_integral_parametric2}
	\feyk \Gamma&=  \Gamma(-d_\Gamma)  \left( \prod_{e\in E_\Gamma} \int \limits_0^\infty \frac{\d a_e\; a_e^{\nu_e-1}}{\Gamma(\nu_e)}  \right)  \; \delta \left( 1-\sum_{e=1}^{\abs{E_\Gamma}} a_e \right) \frac{\phi^{ d_\Gamma}}{ \psi^{d_\Gamma+\frac D 2}} .
\end{align}
Here, $d_\Gamma :=  \abs{L_\Gamma} \frac D 2-\sum_{e\in E_\Gamma} \nu_e $ is the superficial degree of divergence of the graph $\Gamma$, and the  gamma function factor $\Gamma(-d_\Gamma)$ represents the superficial divergence of the integral. For a IR-finite graph without UV-subdivergences, assuming a suitable choice of external momenta, the remaining projective integral is finite. In particular, the form $\alpha_\Gamma$ (\cref{def:omega}) used in the main text is well-defined on the projective space of Schwinger parameters \cite{gaiotto_higher_2024,wang_factorization_2024}.

In general, a projective integral of the form \cref{feynman_integral_parametric2}  might involve additional singularities, corresponding to subdivergences of the Feynman graph. \Cref{feynman_integral_parametric2} is then the starting point for renormalization in parametric space, see e.g. \cite{binoth_automatized_2000}.

\bibliography{vanishing_wedge_product}

\end{document}